\documentclass[12pt]{article}

\usepackage{amsmath}
\usepackage{amssymb}
\usepackage{amsthm}
\usepackage{tikz}
\usetikzlibrary{decorations.markings,decorations.pathreplacing,calc}
\usetikzlibrary{arrows.meta}
\usepackage{enumerate}
\usepackage{comment}
\usepackage{cite}

\usepackage[margin=1in]{geometry}
\usepackage{subfig}
\usepackage[utf8]{inputenc}

\usepackage[colorlinks = true, citecolor=blue]{hyperref}

\usepackage{array}

\newcommand{\be}{\begin{equation}}
\newcommand{\ee}{\end{equation}}
\newcommand{\ba}{\begin{array}}
\newcommand{\ea}{\end{array}}
\newcommand{\bea}{\begin{eqnarray}}
\newcommand{\eea}{\end{eqnarray}}

\newenvironment{acknowledgement}

\newtheorem{dfn}{Definition}

\newtheorem{lem}{Lemma}
\newtheorem*{lem*}{Lemma}
\newtheorem{alg}{Algorithm}
\newtheorem{rem}{Remark}

\newtheorem{conj}{Conjecture}
\newtheorem{thm}{Theorem}
\newtheorem*{thm*}{Theorem}

\newcommand{\footremember}[2]{%
    \footnote{#2}
    \newcounter{#1}
    \setcounter{#1}{\value{footnote}}%
}

\newcolumntype{C}[1]{%
 >{\vbox to 4ex\bgroup\vfill\centering}%
 p{#1}%
 <{\egroup}}  

\begin{document}

\title{Quantum supremacy and random circuits}
\author{Ramis Movassagh  \footremember{ibmcambridge}{IBM Research, Cambridge, MA 02142, U.S.A., Email: ramis@us.ibm.com}
}

\date{\today}

\maketitle

\begin{abstract}
\noindent As Moore's law reaches its limits, quantum computers are
emerging with the promise of dramatically outperforming classical
computers. We have witnessed the advent of quantum processors with
over $50$ quantum bits (qubits), which
are expected to be beyond the reach of classical simulation~\cite{arute2019quantum,huang2020mermin,dalzell2018many}. 
Quantum supremacy is the event at which the old Extended Church-Turing Thesis is overturned: A quantum computer performs a task that is practically impossible
for any classical (super)computer.
The demonstration requires both a solid theoretical
guarantee and an experimental realization. The lead candidate is Random
Circuit Sampling (RCS), which is the task of sampling from the output
distribution of random quantum circuits. 
Google recently announced a $53-$qubit experimental demonstration
of RCS~\cite{arute2019quantum}.
Soon after, classical algorithms appeared that challenge the supremacy
of random circuits by estimating their outputs ~\cite{napp2019efficient,gray2020hyper,huang2020classical}. How hard is it to classically simulate the output of random quantum circuits? 

We prove that estimating the output 
probabilities 
of random quantum circuits is formidably hard ($\#P$-Hard) for any classical computer.  This makes RCS the strongest candidate for demonstrating quantum supremacy 
relative to all other proposals. The robustness to the estimation error that we prove 
may serve as a new hardness criterion for the performance of classical algorithms.
%
To achieve this, we introduce the Cayley path interpolation
between 
any two gates of a quantum computation
and convolve recent advances in quantum
complexity and information with probability and random matrices.
Furthermore, we apply algebraic geometry to generalize the well-known Berlekamp-Welch 
algorithm that is widely used in coding theory and cryptography.
Our results imply that there is an exponential hardness barrier for the classical simulation of most quantum circuits.

\end{abstract}
\maketitle

\section{\label{sec:Overview}Introduction}

Quantum computation is the only model of computation that might solve
certain computational tasks exponentially faster than any standard
supercomputer. The excitement is compounded by the fact that the Moore's
law is reaching its limits and quantum phenomena are relevant if
the chips are to be made smaller. First proposed by Richard Feynman
for the efficient simulation of (quantum) matter, quantum computing
since has rapidly advanced with novel quantum algorithms and tantalizing
prospects. Examples include exponential speed-ups
for integer factorization~\cite{shor1999polynomial}, solution of
linear systems~\cite{harrow2009quantum}, and algorithms based on
quantum walks~\cite{childs2009universal}. From a theoretical standpoint,
the exponential separation in the computational power of quantum computers
would refute the Extended Church-Turing Thesis (ECTT), which asserts that
a probabilistic Turing machine can efficiently simulate any realistic
model of computation. 

Given the state of affairs, we cannot unconditionally prove that quantum
mechanics is hard to classically simulate without making plausible
complexity theoretical assumptions
(e.g.,~$P\ne NP$ or non-collapse of the polynomial hierarchy).
These assumptions pre-date quantum computing and are believed
to be valid foundational bases in computer science.

The algorithms such as factoring seem to require fault-tolerant quantum
computation, which despite the current vigorous efforts is still a
distant goal. However, ``Noisy Intermediate Scale Quantum (NISQ)''\cite{preskill2018quantum}
computers have arrived with about $54$ high fidelity quantum bits
(qubits)~\cite{arute2019quantum,huang2020mermin}. Currently there
is a large global interest with an unprecedented academic and industrial
push (e.g.,~from IBM and Google) for developing these computers to
scale up while keeping quantum coherence. Ultimately one hopes to have enough
good qubits to reify quantum error correcting codes and perform fault-tolerant
quantum computation. In the meantime, the central question in quantum
computing is: In the absence of fault tolerance what can the quantum
computers do? Recent progress shows a mild separation between classical and NISQ computers {\cite{bravyi2018quantum,bravyi2020quantum} yet other work point out obstacles in variational algorithms designed for NISQ computers~\cite{bravyi2019classical, bravyi2019obstacles}.  A milestone is to prove a large separation between the power of NISQ computers and classical ones with minimal number of
complexity theoretical assumptions. This event has been termed quantum
supremacy~\cite{preskill2018quantum}: An event when a NISQ computer
efficiently performs a computational task which would take a formidably long
time (e.g.,~exponential in $n$) to perform on any classical computer.
Although the task may not be of practical use, if demonstrated successfully,
it would show that quantum computers indeed have awesome computational
powers and, more fundamentally, that quantum mechanics is hard to
simulate classically-- whereby refuting ECTT. 

The modern quantum supremacy proposals are based on sampling problems,
which have the advantage that they need little extra assumptions and
can now be performed experimentally on NISQ computers~\cite{arute2019quantum}. Aaronson showed that refutation of ECTT for sampling also refutes ECTT for search problems~\cite{aaronson2014equivalence}. The sampling
task aims to show that outputting samples from a distribution that
mimics the distribution of the quantum process (e.g.,~circuit) is
classically very hard. Moreover, the hardness
of sampling relies on complexity theoretical assumptions that pre-date
quantum computing and do not appeal to quantum mechanics for their
justification (see~\cite{harrow2017quantum,lund2017quantum} for
reviews). The main assumption is the non-collapse of the polynomial
hierarchy, which is stronger than $P\ne NP$ but is as plausible.
This assumption is far less restrictive than those needed for justifying
that, say integer factoring, is hard. The latter needs to assume that
integer factoring is hard classically, which is not known. Indeed,
because of these reasons Google focused on random circuit sampling (RCS)
to experimentally demonstrate quantum supremacy~\cite{arute2019quantum}.

But why random circuits? Complexity theoretical statements are {\it worst-case}
statements. For example, if one says that a problem is $NP-$Hard 
it means that it is $NP-$Hard in general, yet most instances may be easy.
From an experimental perspective, such hardness statements imply that
there exists circuits that are hard to simulate on a classical
computer (unless $P=NP$), but how can one make experimental progress on this claim?
What circuit to pick and how to verify hardness? What is needed
is that the problem is not just hard for some instance
rather it possesses \textit{average-case hardness}. Average-case hardness
is implied if most circuits are hard, which in turn implies that a
random circuit is hard with high probability. Hence, if a random choice
of the quantum circuit is implemented in the lab, then certified average-case
hardness  provides provable guarantees that
the task is likely hard classically. And a (large) family of experiments can be
performed that are hard to simulate classically. 
The caveat is that the average-case
hardness is much stronger than worst-case; therefore till now it was
only conjectured to be true in all quantum supremacy proposals. 

The average-case hardness has been a long standing open conjecture.
In this paper we prove that computing the output probabilities of most quantum circuits is $\#P$-Hard \footnote{$\#P$ is a generalization of NP that extends the decision problems
to counting problems. Informally an NP-complete problem asks:
Does a 3-SAT instance have a solution? The answer is no (i.e.,~zero
solution) or yes (at least one solution). Whereas, $\#P$ asks: How
many solutions does a 3-SAT instance have?}, which implies average-case hardness of computing the probability amplitudes on a classical computer unless the polynomial hierarchy collapses. 
This
indeed makes RCS the strongest candidate
for the demonstration of quantum supremacy in the near-term quantum computing era as it provides strong complexity theoretical evidence for the hardness of sampling. 
We then prove that even approximating the probabilities to small additive
errors is hard; this is referred to as robustness. The robustness
we prove is fully quantifiable and can be benchmarked against the
experiments and the performance of classical algorithms. These bounds
on noise resilience for a grid of $\sqrt{n}\times\sqrt{n}$ qubits
leads to hardness of computing probabilities to within the additive
error of $2^{-O(n^{3})}$ for circuit of depth $\sqrt{n}$, and $2^{-O(n^{2})}$
for constant depth circuits. It is noteworthy that the very recent
numerical algorithms for simulating general constant depth quantum
circuits~\cite{napp2019efficient,bravyi2019classical} take exponential
time if the approximation error falls within our bounds. 
In particular, Napp \textit{et al} numerically
simulate random instances of a universal family of constant depth circuits to $2^{-O(n)}$ additive error~\cite{napp2019efficient}.
A recent paper from AliBaba group makes a similar claim for larger depth circuits such as those used in the Google experiment~\cite{huang2020classical}. This challenges the quantum supremacy experiments as classical simulation of  a random  quantum circuit of depth $\sqrt{n}$ to within $2^{-O(n)}$ additive error is supposed to imply hardness of sampling. 
Therefore, our result may serve
as a new and provable quantum supremacy criterion, albeit a stringent one, for classical algorithms that simulate the output of a quantum circuit.

To accomplish this we provide an explicit and efficient construction
of one-parameter family of unitary matrices (i.e.,~quantum gates)
that interpolate between any two fixed unitaries by varying an interpolation
parameter, $\theta$, between zero and one. The path is based on the
Cayley transformation which maps between hermitian and unitary matrices. Using random matrix and probability theories we prove that
the distribution of each scrambled gate is arbitrarily close to the
Haar measure with a total variation distance that vanishes with increasing
$n$. One says that a quantum circuit is an instance of an average
case circuit if its (local) computational gates enact random Haar
unitaries. Therefore, our interpolation scheme may enable the study
of transitions into quantum chaos or area-to-volume law entanglement
entropies in the study of black holes~\cite{hayden2007black}, and
holographic models of quantum gravity~\cite{takayanagi2018holographic}.
Lastly, we prove an extension of Berlekamp-Welch (BW)~\cite{welch1986error}
algorithm that can efficiently interpolate rational functions (Alg.~\ref{alg:BW_Manuscript}). The discovery of BW algorithm was originally motivated
by classical error correction schemes such as the Reed-Solomon codes
\cite{reed1960polynomial}. In these codes, the messages are encoded
in the coefficients of polynomials over finite fields. BW is remarkable
in that it exactly recovers such polynomials even if the evaluation
of the polynomial at some number of points is erroneous. Our Alg.~\ref{alg:BW_Manuscript} extends the possibility of message encodings
to more general classes of functions.

\subsection{Previous work}
\begin{table}
\centering{}\includegraphics[scale=0.50]{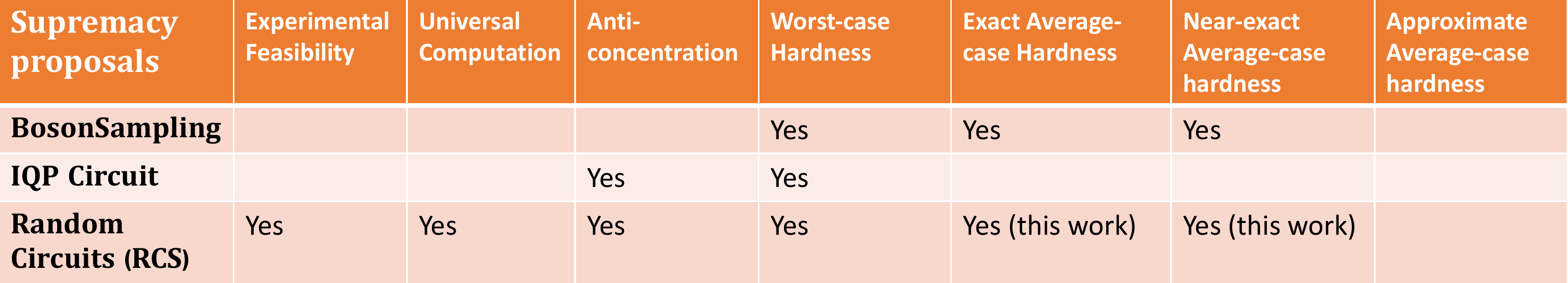}\caption{\label{Table:Summary}Main supremacy proposals, where by``near-exact" vs. ``approximate" we mean to within additive errors of $2^{-poly(n)}$ vs. $2^{-O(n)}/poly(n)$ respectively. See the main Theorem.}
\end{table}
The sampling-based quantum supremacy proposals originate from the
seminal work of Aaronson and Arkhipov~\cite{aaronson2011computational}.
There are three main candidates for demonstrating quantum supremacy
in the near term. These are BosonSampling~\cite{aaronson2011computational},
instantaneous model of quantum computation (known as IQP circuits)
\cite{bremner2011classical,bremner2016achieving}, and RCS~\cite{arute2019quantum,boixo2018characterizing,bouland2018quantum}.
RCS is the encompassing candidate because it relies on minimum
number of complexity theoretical assumptions, it
is amenable to experimental tests of larger size that outstretch the
limit of classical simulations, and is the only one of the three that implements universal quantum computation. For example, in BosonSampling experiments,
it is very hard to scale the number of the photons anywhere close
to what is needed to demonstrate supremacy~\cite{neville2017classical}.
From a foundational vista, BosonSampling needs to make two complexity
theoretical assumptions: 1. Multiplicative estimation of permanents
of Gaussian random matrices is $\#P$-Hard. 2. Anti-concentration conjecture holds, which assumes that the permanent
of random gaussian matrices is not too concentrated near zero. The
latter is crucial for arguing that the additive estimation of the
permanent is equivalent to the multiplicative approximation. The IQP
circuits have the added advantage that they can be implemented on
quantum circuit hardware and require one main extra assumption, which
is the average-case hardness. This assumes that the complexity of
random IQP coincides with the most contrived (i.e.,~the worst-case
hard instances) and is analogous to the first assumption in BosonSampling.
As discussed above, assumption of average-case hardness is quite strong
and requires a proof for its justification. Further, IQP circuits
are restrictive models of computation because their gates are chosen
from a small discrete set of possible gates. In fact, the circuits
are essentially classical (diagonal in $z-$basis) except at the first
and last layer of computation. This makes IQP circuits non-universal
for quantum computation yet very interesting from a quantum complexity
perspective.

RCS is based on random operations, which makes the underlying circuit universal and the most general 
instantiation of average-case computation. The average-case
instances are minimally constrained and the local gates can enact
any realization of $SU(2)$ or $SU(4)$. Moreover, rigorous
anti-concentration bounds have been proved for RCS~\cite{harrow2018approximate},
which show that the output distribution is close to the uniform. This
in turn implies that the signal to noise ratio of probabilities is
relatively higher. We summarize the known key results for the main supremacy proposals in Table \ref{Table:Summary}.

The only other theoretical result applicable to the hardness of RCS
is the nice work of Bouland \textit{et al}~\cite{bouland2018quantum},
which proves a hardness result of computing exact amplitudes of a
\textit{non-unitary} approximation of the actual circuit (see Subsection~\ref{subsec:Inadequacy-of-Taylor} and also Appendix A in~\cite{napp2019efficient}).
This requires a new complexity theoretical assumption, which is the
existence of a classical algorithm that takes as inputs non-unitary
``circuits'' and efficiently produces the probability amplitudes.
As shown in Subsection~\ref{subsec:Inadequacy-of-Taylor}, the non-unitary
oracles will not yield a realistic robustness on actual quantum circuits
and hardware. The latter point was furthered emphasized in the work
that followed ours~\cite{napp2019efficient}. Moreover, due to the
non-unitary approximation, the closeness to random unitary circuits
was not proved.

This work gives an entirely new construction and proof of average-case hardness that is free
of the previous limitations. It also proves quantifiable robustness bounds which can be 
compared with experiments and new classical algorithms that have
challenged the claim. Because of this work, RCS now has the fewest assumptions and is the strongest candidate for demonstration of quantum supremacy.
Our proof simply asserts that estimating the probability amplitudes, even on average (random quantum circuits), is  $\#P$-Hard for any classical supercomputer. For ECTT to be refuted (a watershed event), our stringent estimation bound needs to be weakened to $2^{-O(n)}$. However, the aforementioned recent numerical algorithms cast  doubt on the possibility of this historical upset.

\subsection{Cayley path and quantum supremacy}

An $n-$qubit quantum computation initializes each of the qubits to
the state $|0\rangle$, which makes the initial joint quantum state
of the circuit $|0^{n}\rangle\equiv(1,0,\dots,0)^{T}$. The initial
state is simply the $2^{n}$ dimensional standard basis vector. A
quantum computation is simply the application of a unitary matrix
to $|0^{n}\rangle$ (i.e.,~a rotation in the Hilbert space). The details
of the unitary matrix are fixed by the quantum algorithm being implemented.
This unitary is instantiated in the lab by a circuit $C$ with some
architecture $\mathcal{A}$ (Fig.~\ref{fig:Architecture_C}). The
final state of the quantum computation is the vector $|\psi\rangle=C\:|0^{n}\rangle$,
where for simplicity of notation $C$ also denotes the $2^{n}$ dimensional
unitary matrix that the circuit instantiates. The square of the absolute
value of any entry of $\psi$ is the probability of occurrence of
that particular outcome upon measurement.
Good quantum algorithms output the answers to the desired computational
tasks with high probability and are more efficient than classical computers. The architecture
$\mathcal{A}$ serves as a blue-print that specifies the connectivity
and interaction of the qubits in the course of the computation but
is otherwise agnostic to the details of the computation. Constrained
by the difficulties in the experiment, $\mathcal{A}$
almost always allows for $1-$ or $2-$qubit operations (depicted
in Fig.~\ref{fig:Architecture_C}). For details and proofs
of what follows see Supplementary Material (SM).
\begin{figure}
\centering{}\includegraphics[scale=0.45]{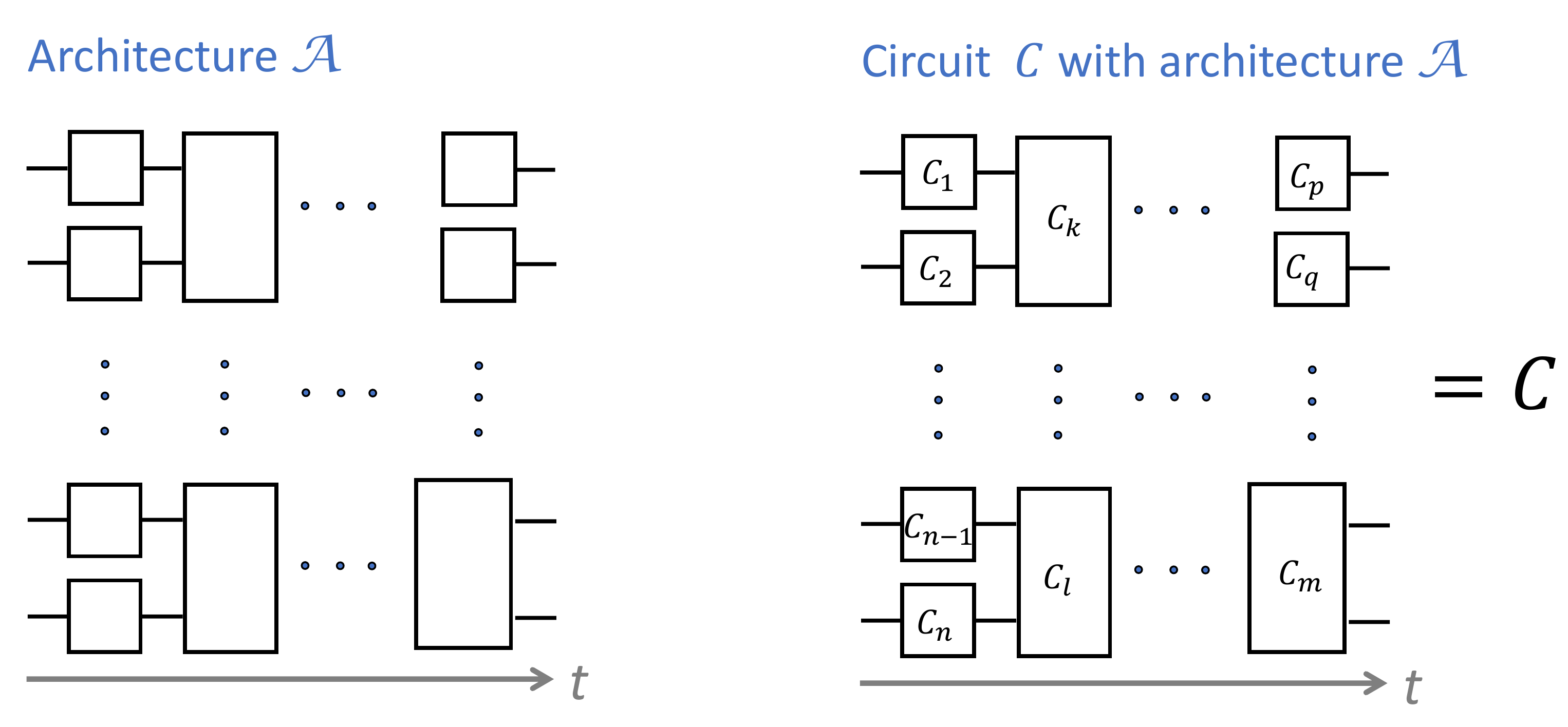}\caption{\label{fig:Architecture_C}Left: The architecture $\mathcal{A}$ is the blue print. Right: Circuit $C$ with architecture
$\mathcal{A}$}
\end{figure}

Suppose $C=\mathcal{C}_{m}\,\mathcal{C}_{m-1}\cdots\mathcal{C}_{1}$
is a quantum circuit with an architecture $\mathcal{A}$ acting on
$n$ qubits with $m$ local unitary gates.
Further, $\mathcal{C}_{k}\equiv C_{k}\otimes I$, which
means that $C_{k}$ enacts 
a non-trivial $1-$ or $2-$qubit operation
and acts trivially on the rest of the qubits as dictated by the architecture
$\mathcal{A}$ (Fig.~\ref{fig:Architecture_C}). By a random circuit
we shall mean that each local unitary is drawn independently and uniformly
at random from the set of all possible unitaries of the appropriate
size (i.e.,~Haar measure). Once the $m$ gates are picked
at random one can sample from the output distribution of the circuit
by running the circuit and measuring the outputs to obtain a string
$x\in\{0,1\}^{n}$ with the probability $p_{x}=|\langle x|C|0^{n}\rangle|^{2}$.
This task is efficiently for a quantum computer as one simply runs
the device and measures the output. Quantum supremacy claims that
to draw strings from a distribution that mimics the probability distribution
induced by the random circuit is computationally hard for any classical
algorithm that takes as input the classical description of the gates.

If sampling from the output of the quantum computer were efficient
then via a well-known algorithm due to Stockmeyer for approximate counting, one could efficiently
estimate the probability amplitudes $p_{x}=|\langle x|C|0^{n}\rangle|^2$
to within a small relative error~\cite{stockmeyer1985approximation}. That is
calculate a quantity $\tilde{p}_{x}$ such that $e^{-\epsilon}p_{x}\le\tilde{p}_{x}\le e^{\epsilon}p_{x}$
for small $\epsilon=O(1/\text{poly}(n))$.  
However, a more natural measure in sampling problems is
estimations with respect to additive errors; that is $p_{x}\pm\epsilon$~\cite{bravyi2016improved,harrow2017quantum}. 
An
attractive feature of RCS is that the output ``anti-concentrates" for circuits of depth $O(\sqrt{n})$~\cite{harrow2018approximate},
which makes the two types of errors equivalent in the sense that $O(1/\text{poly}(n))$ relative error estimation implies an additive error estimation of $p_{x}\pm 2^{-n}/\text{poly}(n)$.
Therefore, to refute ECTT it is
sufficient to prove that estimating $p_{x}$ for random circuits to $2^{-n}/\text{poly}(n)$ additive error 
is hard for any classical algorithm. Moreover, because of a property
known as ``hiding''~\cite{aaronson2011computational} it is sufficient
to prove the average-case $\#P-$Hardness of
\begin{equation}
p_{0}\equiv|\langle0^{n}|C|0^{n}\rangle|^{2}\quad,\label{eq:p0-1}
\end{equation}
to inverse polynomial relative error.
\begin{figure}
\centering{}\includegraphics[scale=0.5]{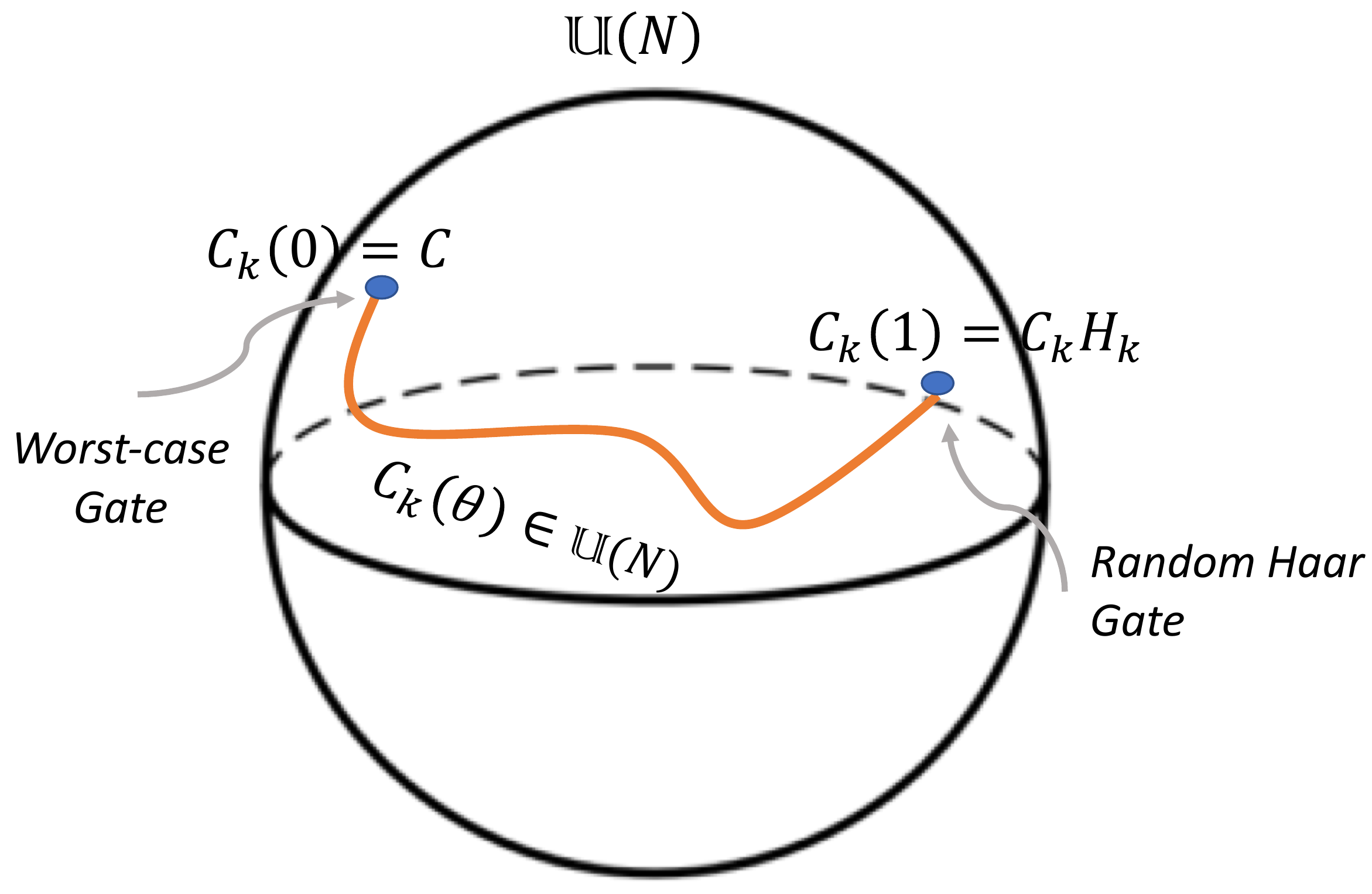}\caption{\label{fig:UnitaryG-Intro} Schematics of the Cayley path on the unitary
group induced by $C_{k}(\theta)\equiv C_{k}f(\theta h_{k})$.}
\end{figure}

It is known that there exists \textit{worst} case quantum circuits
with a specific architecture ${\cal A}$ for which sampling is $\#P$-Hard.
In seminal works Terhal and DiVincenzo~\cite{terhal2002adaptive}
proved the existence of such circuits for constant depth (four) circuits
and Bremner\textit{ et al} proved it for IQP circuits with depth of $\sqrt{n}$
\cite{bremner2016achieving}. We also know that these worst-case
instances are hard to within a constant relative error\cite{dyer2004relative}.
The distribution over the circuits with the same architecture and
Haar local unitaries is denoted by ${\cal H}_{{\cal A}}$. Can we
mathematically make the $\#P$-Hardness of the worst-case
to average-case circuits with the same architecture equivalent? Suppose $C_{1},\dots,C_{m}$
are the gates of the worst case circuit with the architecture $\mathcal{A}$,
and $H_{1},\dots,H_{m}$ are corresponding Haar random gates, where instead of each $C_{k}$ in Fig.~\ref{fig:Architecture_C},
one puts $H_{k}$. The latter is an average-case instance of $\mathcal{H}_{\mathcal{A}}$.
We propose a continuous path, that we call the Cayley path, which
connects any two circuits of the same architecture. The Cayley function
is defined by $f(x)=(1+ix)/(1-ix)$ where $i=\sqrt{-1}$, where one
defines $f(-\infty)=-1$. This function maps any real $x$ to a unique
point on a unit circle in the complex plane (SM). Given a local gate $H_{k}$
we can always write it as $H_{k}=f(h_{k})$ where $h_{k}$ is a Hermitian
matrix and the Cayley function maps the real eigenvalues of $h_{k}$
to the complex eigenvalues of the unitary matrix $H_{k}$; the eigenvectors
remain the same. We define a path parametrized by a real valued $\theta$
that interpolates between $C_{k}$ and $H_{k}$ and is a unitary for
any $\theta$ (Fig.~\ref{fig:UnitaryG-Intro})
\begin{equation}
C_{k}(\theta)=C_{k}\:f(\theta\:h_{k})\quad,\label{eq:CayleyPath}
\end{equation}
where we recall that each $C_{k},H_{k}$ is an $N\times N$ unitary
matrix with $N\in\{2,4\}$. Note that $C_{k}(0)=C_{k}$ is the gate
of the worst-case instance and by the translation invariance of Haar
measure $C(1)=C_{k}H_{k}$ is a random Haar unitary (i.e.,~
an average-case instance).
The entries of $C_{k}(\theta)$ are rational functions
of degree $(N,N)$, i.e.,~the numerator and denominator are polynomials
of degree $N$ because of the algebraic form of the Cayley function (SM). Under the Cayley path the whole circuit which enacts
a $2^{n}\times2^{n}$ matrix transforms gate-by-gate as
\[
C(\theta)=\mathcal{C}_{m}(\theta)\,\mathcal{C}_{m-1}(\theta)\cdots\mathcal{C}_{1}(\theta)\quad;
\]
consequently each entries of $C(\theta)$ is a rational function of
degree $(mN,mN)$. In Lemma \eqref{lem:TVD_Haar} of SM we show that
under Cayley path parametrization (Eq.~\eqref{eq:CayleyPath})
the distribution over the circuits $C(\theta)$ for $|\theta-1|\le\Delta\ll1$
is 
arbitrarily close to $\mathcal{H}_{\mathcal{A}}$ in total variation
distance. This implies that the circuits are indeed generic instances
of the average-case
 for $\theta\rightarrow1$.

Since $\langle0^{n}|C(\theta)|0^{n}\rangle$ is the $(1,1)$ entry
of the matrix $C(\theta)$, the quantity of interest in Eq.~\eqref{eq:p0-1}
is a rational function of degree $(2mN,2mN)$ in $\theta$
\begin{eqnarray}
p_{0}(\theta) & \equiv & |\langle0^{n}|C(\theta)|0^{n}\rangle|^{2}\quad,\label{eq:p0_theta-1}
\end{eqnarray}
and our task is to prove under plausible complexity theoretical assumptions
that approximating this quantity to small additive errors is $\#P-$Hard
for \textit{any} classical algorithm. 
\begin{figure}
\includegraphics[scale=0.5]{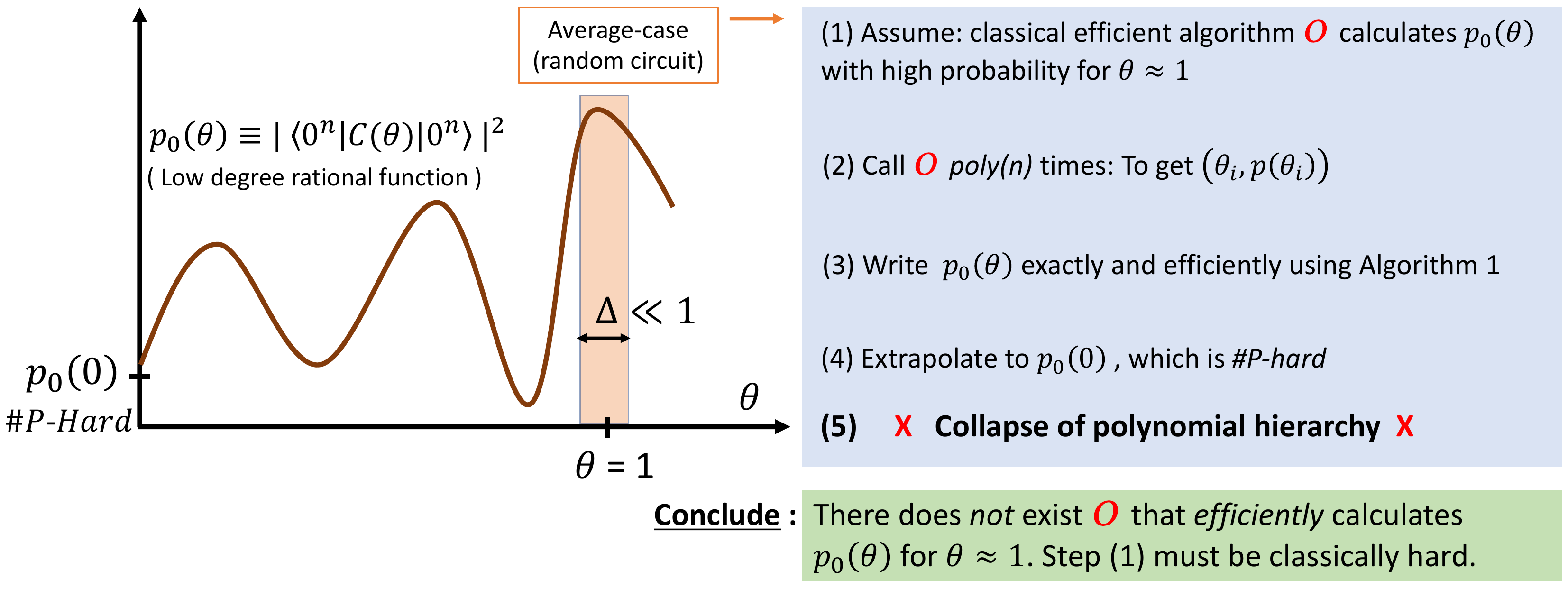}\caption{\label{fig:Reduction-idea}The worst- to average-case hardness reduction for a
circuit with architecture $\mathcal{A}$}
\end{figure}

We prove this by showing that if this task were efficient on a classical
computer then it would violate the widely believed
assumption of the non-collapse of the polynomial hierarchy (Fig.~\ref{fig:Reduction-idea}). Our main theorem is (see Theorems \eqref{thm:MainResult-1}
and \eqref{thm:Robustness} in SM):\\

\textbf{Theorem:} 
Suppose there exists an architecture $\mathcal{A}$ for which it is $\#P-$Hard to compute arbitrary output probabilities 
to within a small multiplicative
error, then it is $\#P-$Hard to calculate the probability amplitudes
for most random circuits with the same architecture $\mathcal{A}$
to within
 $\epsilon=2^{-\Omega(m^2)}$ additive error, where $m$ is the number of gates.\\

The theorem states that it is hard to calculate the amplitudes for
\textit{most} circuits. Indeed a random circuit may accidentally coincide
with the 'worst-case' circuit and such a scenario has to be excluded. More practically, the requirement
that the algorithm is hard for most instances leads to beautiful connections
with coding theory and algebraic geometry. In order to prove this
statement, one first assumes that there exists a classical algorithm
that efficiently computes $p_{0}(\theta)$ for $\theta\approx1$
(i.e.,~average-case circuits), then by calling the algorithm on $\text{poly}(n)$
number of points $\theta_{i}$, one can explicitly solve for the
rational function that defines $p_{0}(\theta)$ for all $\theta$.
This would be easy to argue if the classical algorithm succeeded in outputting $p(\theta_{i})$ on
any input $\theta_{i}\approx1$. But
we can only assume that it succeeds on some inputs. A
remarkable algorithm due to Berlekamp and Welch succeeds in explicitly
specifying  a polynomial if the rate of error (number of inputs on
which the algorithm fails) is low enough~\cite{welch1986error}.
The intuition behind the algorithm is that two polynomials with sufficient
number of intersections must be the same polynomial. This is also
true for more general algebraic curves as implied by Bezout's theorem
in algebraic geometry; we generalize BW algorithm to rational functions
and prove the following Algorithm (Alg. \eqref{alg:(Berlekamp-Welch-for-Rational}
in SM):
\begin{alg}\label{alg:BW_Manuscript}
Given $(\theta_{1},f_{1})$, $(\theta_{2},f_{2})$,
..., $(\theta_{n},f_{n})$, output the rational function $F(\theta)$
of degree $(k_{1},k_{2})$ by evaluating it at $n>k_{1}+k_{2}+2t$
points despite $t$ errors in the evaluation points.
\end{alg}
The standard  BW algorithm is used for encoding messages
in the coefficients of polynomials and allows for reliable transmission
despite errors. Our algorithm generalizes this possibility to rational
functions.

The second part of the theorem proves that not only $p_{0}$ is hard
for classical computers but that the hardness is robust and even computing
the approximation $p_{0}\pm\epsilon$ is $\#P$-Hard. Suppose the circuit has $n$ qubits and has a two dimensional
architecture on a $\sqrt{n}\times\sqrt{n}$ grid, which is the common current
architecture with superconducting qubits (e.g.,~at IBM
and Google). If the depth is a constant
then the number of local gates is $m=O(n)$ and if the depth is $d=O(\sqrt{n})$
then $m=O(n^{3/2})$. For these we prove robustness with respect to additive errors of $\epsilon=2^{-O(n^{2})}$
for constant depth circuits, and $\epsilon=2^{-O(n^{3})}$
for Google's experiment whose depth is $\sqrt{n}$ (SM). 

As stated in the introduction, for constant depth circuits, classical
numerical algorithms seem to exhibit a hardness phase transition with
respect to the additive error. For example, it is seen
that the classical algorithm finds it hard to provide an additive
error approximations that we prove, yet seems to be able to efficiently
simulate universal circuits for $\epsilon=2^{-n}/\text{poly}(n)$~\cite{napp2019efficient}.
For circuits with depth $\sqrt{n}$ we proved hardness
of estimating probabilities with respect to $2^{-O(n^{3})}$ additive
error, whereas AliBaba challenges the original supremacy proposal by claiming to simulate these circuits up to the error of $2^{-O(n)}$~\cite{huang2020classical}. Since in the near-term the number of qubits
is approximately $n\sim100$, the exact quantification might prove
helpful in practice.
If it does happen that deeper circuits
are hard with respect to $2^{-n}/\text{poly}(n)$ additive error,
then any theoretical guarantee would most likely need to use proof
techniques that are not based on extrapolations of low degree algebraic
functions. More research in this direction
is needed to better chart the topography of the computation power
of near-term quantum computers.

In this work we provided the strongest rigorous hardness results for
quantum supremacy proposals to date applicable to the recent experimental
breakthroughs~\cite{arute2019quantum}. In doing so we offer new tools
for interpolating between circuits that may help elucidate the computational
power of quantum circuits as a function of their architecture. Further,
we provide an algorithm for rational function extrapolation that
may be useful in fresh new contexts such as coding theory. Future
work may include investigation of the optimality of our robustness bounds especially
that the numerical algorithms seem to reach $2^{-O(n)}$ additive error. It would be interesting to (numerically) investigate the phase transition with respect to the additive error 
for deeper circuits (e.g.,~depth $\sqrt{n}$). Lastly, we envision the utility of Cayley path scrambling of the quantum
circuits in the study of models of quantum gravity and
other contexts such as in cryptography, circuit hiding, blind quantum
computation, quantum computation by (extra)interpolation, and quantification of the power of quantum circuits as a function of their architecture.

\begin{acknowledgement}
We thank Sergey Bravyi, John Napp, Alex Dalzell, Sergio Boixo, and Jeffrey Schenker for discussions. We also thank Adam Bouland, Bill Fefferman, Yunchao Liu, and Karol Zyczkowski.
I am grateful for the support of the Frontiers Foundation and the MIT-IBM collaborative
grant.
\end{acknowledgement}


\bibliographystyle{alpha}

\bibliography{mybib}

\newcommand{\etalchar}[1]{$^{#1}$}
\begin{thebibliography}{NLPD{\etalchar{+}}20}

\bibitem[AA11]{aaronson2011computational}
Scott Aaronson and Alex Arkhipov.
\newblock The computational complexity of linear optics.
\newblock In {\em Proceedings of the forty-third annual ACM symposium on Theory
  of computing}, pages 333--342. ACM, 2011.

\bibitem[AAB{\etalchar{+}}19]{arute2019quantum}
Frank Arute, Kunal Arya, Ryan Babbush, Dave Bacon, Joseph~C Bardin, Rami
  Barends, Rupak Biswas, Sergio Boixo, Fernando~GSL Brandao, David~A Buell,
  et~al.
\newblock Quantum supremacy using a programmable superconducting processor.
\newblock {\em Nature}, 574(7779):505--510, 2019.

\bibitem[Aar14]{aaronson2014equivalence}
Scott Aaronson.
\newblock The equivalence of sampling and searching.
\newblock {\em Theory of Computing Systems}, 55(2):281--298, 2014.

\bibitem[BFNV19]{bouland2018quantum}
Adam Bouland, Bill Fefferman, Chinmay Nirkhe, and Umesh Vazirani.
\newblock On the complexity and verification of quantum random circuit
  sampling.
\newblock {\em Nature Physics}, 15(2):159, 2019.

\bibitem[BG16]{bravyi2016improved}
Sergey Bravyi and David Gosset.
\newblock Improved classical simulation of quantum circuits dominated by
  clifford gates.
\newblock {\em Physical review letters}, 116(25):250501, 2016.

\bibitem[BGK18]{bravyi2018quantum}
Sergey Bravyi, David Gosset, and Robert K{\"o}nig.
\newblock Quantum advantage with shallow circuits.
\newblock {\em Science}, 362(6412):308--311, 2018.

\bibitem[BGKT20]{bravyi2020quantum}
Sergey Bravyi, David Gosset, Robert Koenig, and Marco Tomamichel.
\newblock Quantum advantage with noisy shallow circuits.
\newblock {\em Nature Physics}, pages 1--6, 2020.

\bibitem[BGM19]{bravyi2019classical}
Sergey Bravyi, David Gosset, and Ramis Movassagh.
\newblock Classical algorithms for quantum mean values.
\newblock {\em arXiv preprint arXiv:1909.11485}, 2019.

\bibitem[BIS{\etalchar{+}}18]{boixo2018characterizing}
Sergio Boixo, Sergei~V Isakov, Vadim~N Smelyanskiy, Ryan Babbush, Nan Ding,
  Zhang Jiang, Michael~J Bremner, John~M Martinis, and Hartmut Neven.
\newblock Characterizing quantum supremacy in near-term devices.
\newblock {\em Nature Physics}, 14(6):595, 2018.

\bibitem[BJS11]{bremner2011classical}
Michael~J Bremner, Richard Jozsa, and Dan~J Shepherd.
\newblock Classical simulation of commuting quantum computations implies
  collapse of the polynomial hierarchy.
\newblock In {\em Proceedings of the Royal Society of London A: Mathematical,
  Physical and Engineering Sciences}, volume 467, pages 459--472. The Royal
  Society, 2011.

\bibitem[BKKT19]{bravyi2019obstacles}
Sergey Bravyi, Alexander Kliesch, Robert Koenig, and Eugene Tang.
\newblock Obstacles to state preparation and variational optimization from
  symmetry protection.
\newblock {\em arXiv preprint arXiv:1910.08980}, 2019.

\bibitem[BMS16]{bremner2016achieving}
Michael~J Bremner, Ashley Montanaro, and Dan~J Shepherd.
\newblock Achieving quantum supremacy with sparse and noisy commuting quantum
  computations.
\newblock {\em arXiv preprint arXiv:1610.01808}, 2016.

\bibitem[Chi09]{childs2009universal}
Andrew~M Childs.
\newblock Universal computation by quantum walk.
\newblock {\em Physical review letters}, 102(18):180501, 2009.

\bibitem[CPS99]{cai1999hardness}
Jin-Yi Cai, Aduri Pavan, and D~Sivakumar.
\newblock On the hardness of permanent.
\newblock In {\em Annual Symposium on Theoretical Aspects of Computer Science},
  pages 90--99. Springer, 1999.

\bibitem[DGGJ04]{dyer2004relative}
Martin Dyer, Leslie~Ann Goldberg, Catherine Greenhill, and Mark Jerrum.
\newblock The relative complexity of approximate counting problems.
\newblock {\em Algorithmica}, 38(3):471--500, 2004.

\bibitem[DHKLP18]{dalzell2018many}
Alexander~M Dalzell, Aram~W Harrow, Dax~Enshan Koh, and Rolando~L La~Placa.
\newblock How many qubits are needed for quantum computational supremacy?
\newblock {\em arXiv preprint arXiv:1805.05224}, 2018.

\bibitem[GK20]{gray2020hyper}
Johnnie Gray and Stefanos Kourtis.
\newblock Hyper-optimized tensor network contraction.
\newblock {\em arXiv preprint arXiv:2002.01935}, 2020.

\bibitem[GS92]{gemmell1992highly}
Peter Gemmell and Madhu Sudan.
\newblock Highly resilient correctors for polynomials.
\newblock {\em Information processing letters}, 43(4):169--174, 1992.

\bibitem[Haa33]{haar1933massbegriff}
Alfred Haar.
\newblock Der massbegriff in der theorie der kontinuierlichen gruppen.
\newblock {\em Annals of mathematics}, pages 147--169, 1933.

\bibitem[HCC{\etalchar{+}}20]{huang2020mermin}
Wei-Jia Huang, Wei-Chen Chien, Chien-Hung Cho, Che-Chun Huang, Tsung-Wei Huang,
  and Ching-Ray Chang.
\newblock Mermin's inequalities of multiple qubits with orthogonal measurements
  on ibm q 53-qubit system.
\newblock {\em Quantum Engineering}, page e45, 2020.

\bibitem[HHL09]{harrow2009quantum}
Aram~W Harrow, Avinatan Hassidim, and Seth Lloyd.
\newblock Quantum algorithm for linear systems of equations.
\newblock {\em Physical review letters}, 103(15):150502, 2009.

\bibitem[HM17]{harrow2017quantum}
Aram~W Harrow and Ashley Montanaro.
\newblock Quantum computational supremacy.
\newblock {\em Nature}, 549(7671):203, 2017.

\bibitem[HM18]{harrow2018approximate}
Aram Harrow and Saeed Mehraban.
\newblock Approximate unitary $ t $-designs by short random quantum circuits
  using nearest-neighbor and long-range gates.
\newblock {\em arXiv preprint arXiv:1809.06957}, 2018.

\bibitem[HP07]{hayden2007black}
Patrick Hayden and John Preskill.
\newblock Black holes as mirrors: quantum information in random subsystems.
\newblock {\em Journal of High Energy Physics}, 2007(09):120, 2007.

\bibitem[HZN{\etalchar{+}}20]{huang2020classical}
Cupjin Huang, Fang Zhang, Michael Newman, Junjie Cai, Xun Gao, Zhengxiong Tian,
  Junyin Wu, Haihong Xu, Huanjun Yu, Bo~Yuan, et~al.
\newblock Classical simulation of quantum supremacy circuits.
\newblock {\em arXiv preprint arXiv:2005.06787}, 2020.

\bibitem[LBR17]{lund2017quantum}
AP~Lund, Michael~J Bremner, and TC~Ralph.
\newblock Quantum sampling problems, bosonsampling and quantum supremacy.
\newblock {\em npj Quantum Information}, 3(1):15, 2017.

\bibitem[Mov18]{movassagh2018efficient}
Ramis Movassagh.
\newblock Efficient unitary paths and quantum computational supremacy: A proof
  of average-case hardness of random circuit sampling.
\newblock {\em arXiv preprint arXiv:1810.04681}, 2018.

\bibitem[NLPD{\etalchar{+}}20]{napp2019efficient}
John Napp, Rolando~L La~Placa, Alexander~M Dalzell, Fernando~GSL Brandao, and
  Aram~W Harrow.
\newblock Efficient classical simulation of random shallow 2d quantum circuits.
\newblock {\em arXiv preprint arXiv:2001.00021}, 2020.

\bibitem[NSC{\etalchar{+}}17]{neville2017classical}
Alex Neville, Chris Sparrow, Rapha{\"e}l Clifford, Eric Johnston, Patrick~M
  Birchall, Ashley Montanaro, and Anthony Laing.
\newblock Classical boson sampling algorithms with superior performance to
  near-term experiments.
\newblock {\em Nature Physics}, 13(12):1153--1157, 2017.

\bibitem[Pat92]{paturi1992degree}
Ramamohan Paturi.
\newblock On the degree of polynomials that approximate symmetric boolean
  functions (preliminary version).
\newblock In {\em Proceedings of the twenty-fourth annual ACM symposium on
  Theory of computing}, pages 468--474. ACM, 1992.

\bibitem[Pre18]{preskill2018quantum}
John Preskill.
\newblock Quantum computing in the nisq era and beyond.
\newblock {\em arXiv preprint arXiv:1801.00862}, 2018.

\bibitem[Rak07]{rakhmanov2007bounds}
Evguenii~A Rakhmanov.
\newblock Bounds for polynomials with a unit discrete norm.
\newblock {\em Annals of mathematics}, pages 55--88, 2007.

\bibitem[RS60]{reed1960polynomial}
Irving~S Reed and Gustave Solomon.
\newblock Polynomial codes over certain finite fields.
\newblock {\em Journal of the society for industrial and applied mathematics},
  8(2):300--304, 1960.

\bibitem[Sho99]{shor1999polynomial}
Peter~W Shor.
\newblock Polynomial-time algorithms for prime factorization and discrete
  logarithms on a quantum computer.
\newblock {\em SIAM review}, 41(2):303--332, 1999.

\bibitem[Sto85]{stockmeyer1985approximation}
Larry Stockmeyer.
\newblock On approximation algorithms for\# p.
\newblock {\em SIAM Journal on Computing}, 14(4):849--861, 1985.

\bibitem[Tak18]{takayanagi2018holographic}
Tadashi Takayanagi.
\newblock Holographic spacetimes as quantum circuits of path-integrations.
\newblock {\em arXiv preprint arXiv:1808.09072}, 2018.

\bibitem[TD04]{terhal2002adaptive}
Barbara~M Terhal and David~P DiVincenzo.
\newblock Adaptive quantum computation, constant depth quantum circuits and
  arthur-merlin games.
\newblock {\em Quant. Inf. Comp.}, 4(2):134--145, 2004.

\bibitem[WB86]{welch1986error}
Lloyd~R Welch and Elwyn~R Berlekamp.
\newblock Error correction for algebraic block codes, December~30 1986.
\newblock US Patent 4 633 470.

\bibitem[Wey64]{weyl1964symetrie}
Hermann Weyl.
\newblock Sym{\'e}trie et math{\'e}matique moderne.
\newblock 1964.

\end{thebibliography}

\newpage{}

\section{Algebraic unitary path}

\subsection{Cayley path}

Let $\mathbb{U}(N)$ be the set of $N\times N$ unitary matrices and
suppose $U_{0}\in\mathbb{U}(N)$ and $U_{1}\in\mathbb{U}(N)$. How
can one build a parametrized path $U(\theta)$ between them such that
$U(\theta)\in\mathbb{U}(N)$ for all $\theta\in[0,1]$ and $U(0)=U_{0}$
and $U(1)=U_{1}$? 

Previously we gave various paths between $U_{0}$ and $U_{1}$ that
were everywhere contained in the unitary group~\cite{movassagh2018efficient}.
In particular, we gave a new rational function-valued path based on
the QR-factorization~\cite{movassagh2018efficient}. 

Here we consider a new extrapolation based on the Cayley transformation.
Suppose $U_{0},U_{1}\in\mathbb{U}(N)$ are unitary matrices and define
the unitary matrix $H\equiv U_{0}^{\dagger}U_{1}$ . Let $x\in\mathbb{R}$
and $f(x)$ be the \textit{Cayley function}
\begin{figure}
\centering{}\textit{\includegraphics[scale=0.3]{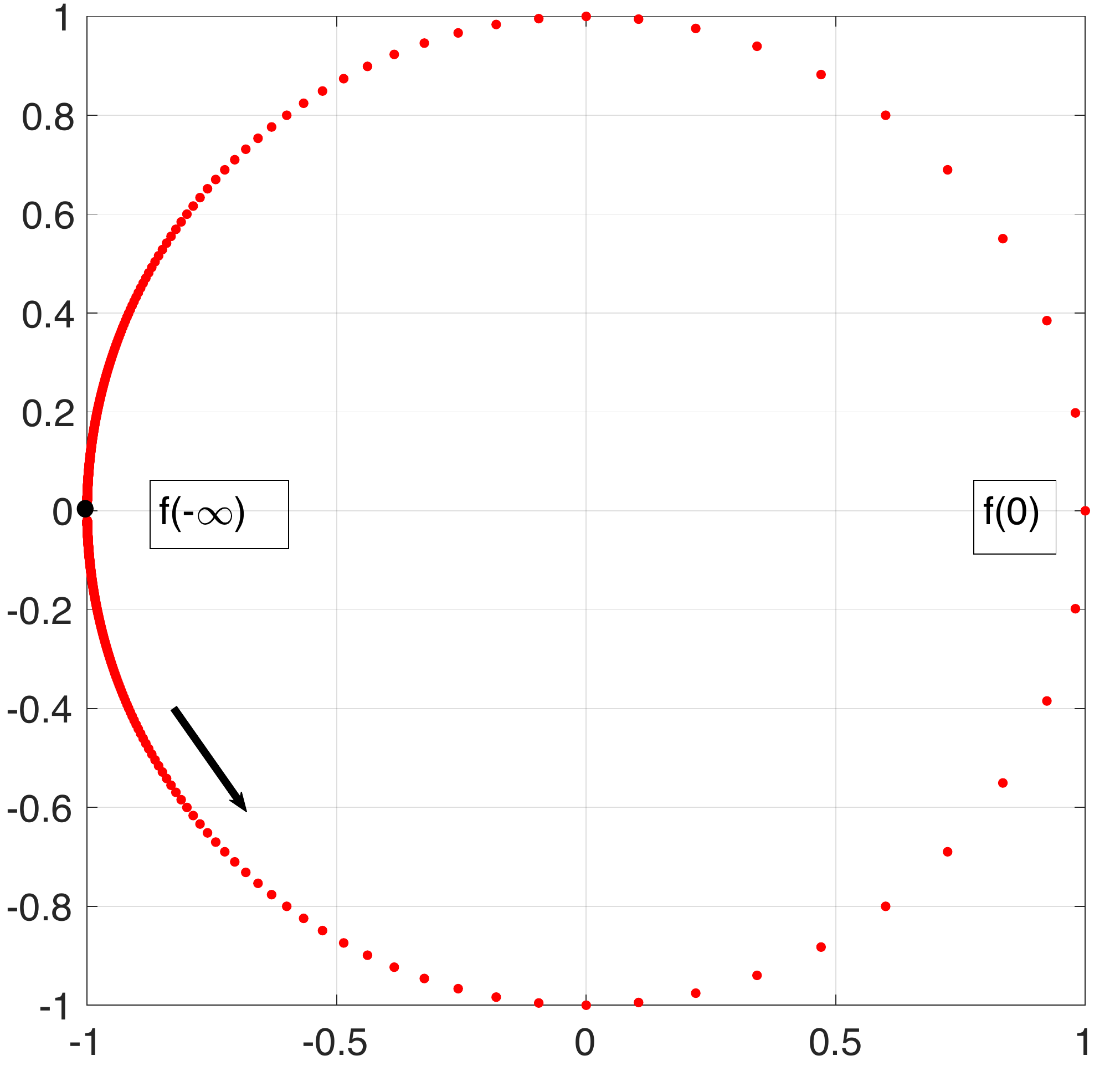}\caption{\textit{\label{fig:Plot-of-Cayley}Plot of the Cayley function in
the complex plane (Eq.~\eqref{eq:f_x}). The arrow shows how the function
fills the unit circle as $x$ increases from $x=-\infty$. The non-uniform
spacing is due to the finite step size in $x$ and aggregation of
points at infinity.}}
} 
\end{figure}
\begin{equation}
f(x)=\frac{1+ix}{1-ix}\label{eq:f_x}
\end{equation}
where we define $f(-\infty)=-1$ (Fig.~\ref{fig:Plot-of-Cayley}).
Since $f(x)$ is a bijection between the real line and the unit circle,
$H$ has the unique representation
\begin{equation}
H=f(h),\quad h=h^{\dagger}\label{eq:CayleyT}
\end{equation}
and it is easy to verify that $H^{\dagger}=f(-h)$.

We want an interpolation $U(\theta)$ such that $U(0)=U_{0}$ and
$U(1)=U_{1}$ and entries of $U(\theta)$ are simple functions of
$\theta$ that can be efficiently computed. Since $h$ is Hermitian,
we have
\[
h=\sum_{\alpha=1}^{N}h_{\alpha}|\psi_{\alpha}\rangle\langle\psi_{\alpha}|,
\]
where $(h_{\alpha},\psi_{\alpha})$ are the eigenpairs of $h$. Since
$h$ is a normal matrix (in fact Hermitian) we can easily express
the spectral decomposition of $H$ as $H\equiv f(h)=\sum_{\alpha=1}^{N}f(h_{\alpha})|\psi_{\alpha}\rangle\langle\psi_{\alpha}|$
.

The proposed path is
\begin{equation}
U(\theta)=U_{0}f(\theta h)=\sum_{\alpha=1}^{N}f(\theta h_{\alpha})\;U_{0}|\psi_{\alpha}\rangle\langle\psi_{\alpha}|.\label{eq:C_THETA}
\end{equation}
$U(\theta)$ is a unitary matrix as it is a product of two unitary
matrices. Note that $U(0)=U_{0}f(0)=U_{0}$ and $U(1)=U_{0}U_{0}^{\dagger}U_{1}=U_{1}$
as desired. We now derive the algebraic dependence of the entries
of $U(\theta)$ on $\theta$.

Using the definition of the Cayley function and foregoing equation
we write
\begin{equation}
U(\theta)=\frac{1}{q(\theta)}\sum_{\alpha=1}^{N}p_{\alpha}(\theta)\text{ }\left(U_{0}\text{ }|\psi_{\alpha}\rangle\langle\psi_{\alpha}|\right),\label{eq:Cayley_path}
\end{equation}
where $U_{0}\text{ }|\psi_{\alpha}\rangle\langle\psi_{\alpha}|$ are
matrices, and $q(\theta)$ and $p_{\alpha}(\theta)$ are the polynomials
of degree $N$ in $\theta$:
\begin{equation}
q(\theta)=\prod_{\alpha=1}^{N}(1-i\theta h_{\alpha})\quad\text{and}\quad p_{\alpha}(\theta)=(1+i\theta h_{\alpha})\prod_{\beta\in[N]\backslash\alpha}(1-i\theta h_{\beta}).\label{eq:q_x}
\end{equation}
So far we took the end-points to be entirely general. 

For the purposes of the proofs pertaining to quantum supremacy we
adapt the above and define the Cayley path on each gate. Let $C_{k}$
be a fixed gate of the worst-case circuit that enacts $N\times N$
unitary and let $H_{k}$ be a random Haar gate of the same size. By
the translation-invariance of the Haar measure $C_{k}H_{k}$ is also
a random Haar gate. We define the Cayley path for each of the gate
of the quantum computation by (see Eq.~\eqref{eq:CayleyPath} in manuscript)
\begin{equation}\label{eq:CayleyPathSM}
C_{k}(\theta)=C_{k}\:f(\theta h_{k})
\end{equation}
where $H_{k}=f(h_{k})$ and $h_{k}^{\dagger}=h_{k}$. Note that $C_{k}(0)=C_{k}$
and $C_{k}(1)=C_{k}H_{k}$ is a random gate.

In the above equation
we make the dependence on $k$ explicit and denote $p_{\alpha}(\theta)\mapsto p_{k,\alpha}(\theta)$
and $q(\theta)\mapsto q_{k}(\theta)$. 
Using the spectral decomposition we write $h_{k}=\sum_{\alpha=1}^{N}h_{k,\alpha}|\psi_{k,\alpha}\rangle\langle\psi_{k,\alpha}|$
where $h_{k,\alpha}$ and $|\psi_{k,\alpha}\rangle$ are the eigenvalues
and eigenvectors of $h_{k}$. We can now
express Eq.~\eqref{eq:CayleyPathSM} (Eq.~\eqref{eq:CayleyPath} of the manuscript) as  $C_{k}(\theta)=q_{k}^{-1}(\theta)\sum_{\alpha=1}^{N}p_{k,\alpha}(\theta)\;C_{k}|\psi_{k,\alpha}\rangle\langle\psi_{k,\alpha}|$
where $q_{k}(\theta)=\prod_{\alpha=1}^{N}(1-i\theta h_{k,\alpha})$
and $p_{k,\alpha}(\theta)=(1+i\theta h_{k,\alpha})\prod_{\beta\in[N]\backslash\alpha}(1-i\theta h_{k,\beta})$

 These
are simply polynomials of degree $N$ that only depend on $\theta$
and $H$. 
\begin{rem}
In Section~\ref{sec:RCS}, we think of $N$ as the size of a local
gate, which is $N=2$ or $N=4$. The entries of $C_{k}(\theta)$ are
rational functions of degree $(N,N)$. However, for a given $\theta$
and $H$, the normalization $q_{k}(\theta)$ is easy to classically
compute. It amounts to a diagonalization of an $N\times N$ matrix
$H$ and an $N-$fold product of the complex numbers $(1-i\theta h_{\beta})$.
Since $N\le4$, this is done in $O(1)$ time. For a general circuit
made up of $m$ gates, the classical computational complexity of calculating
all $q_{k}(\theta)$'s is therefore $O(m)$. By precomputing them
all and multiplying through, Eq.~\eqref{eq:p0-1} is effectively polynomial-valued
and can be treated formally as such. This will be made precise below
and $\prod_{k}q_{k}(\theta)$ will be bounded to guarantee non-divergence
of the rational functions.
\end{rem}
We now turn to the issue of uniquely determining a rational function
by efficient sampling.

\subsection{Berlekamp-Welch for rational functions}

\begin{lem}\label{fact:Rational-function}
Any rational function of degree $(k_{1},k_{2})$
in one variable $\theta$ has the general form
\[
F(\theta)=\frac{a_{k_{1}}\theta^{k_{1}}+a_{k_{1}-1}\theta^{k_{1}-1}+\cdots+a_{0}}{b_{k_{2}}\theta^{k_{2}}+b_{k_{2}-1}\theta^{k_{2}-1}+\cdots+b_{0}}
\]
and is uniquely determined by $k_{1}+k_{2}+1$ points provided that
$F(\theta_{i})=f_{i}<\infty$ for $i\in[k_{1}+k_{2}+1]$ are independent
conditions. 
\end{lem}
\begin{proof}
Since a rational function is determined up to a constant multiple
of numerator and denominator, we can factor out $a_{k_{1}}/b_{k_{2}}$,
whereby the number of unknown coefficients are $k_{1}+k_{2}+1$. By
multiplying both sides by the denominator and then evaluating $F(\theta)$
at $k_{1}+k_{2}+1$ points $F(\theta_{i})=f_{i}$, the coefficients
become the solution of the linear system of equations in $(k_{1}+k_{2}+1)$
variables. Given that the $f_{i}$ are independent, the coefficients
are uniquely determined (unique point of intersection of hyperplanes).
Lastly, $f_{i}<\infty$ is to emphasize that we discard any $\theta_{i}$
that is a root of the denominator; such $\theta$'s are of measure
zero anyway. 
\end{proof}
In coding theory, and especially in Reed-Solomon codes~\cite{reed1960polynomial},
the messages $a_{0},\dots,a_{k}$ may be encoded into a polynomial
$a_{0}+a_{1}\theta+\cdots+a_{k}\theta^{k}$, which then is evaluated
at $n>k+t+1$ points. Then, the decoding procedure recovers the polynomial
and hence the message exactly despite $t$ errors. The decoding procedure
relies on BW algorithm for polynomial interpolation
\cite{welch1986error,gemmell1992highly}. BW can be extended to interpolate
rational functions. The proof follows Sudan's and is a generalization
of it from polynomial to rational functions~\cite{gemmell1992highly}. 
\begin{dfn}
(Error polynomial) \label{Def: ErrorPolynomial}Suppose $f=(f_{1},\dots,f_{n})$
is a vector. Let $F(\theta)$ be a rational function of degree $(k_{1},k_{2})$.
We define the error polynomial $E(\theta)$ as one that satisfies
\[
E(\theta_{i})=0\quad\text{if}\quad F(\theta_{i})\ne f_{i},\quad\text{deg}(E(\theta))\le t.
\]
\end{dfn}

\begin{alg}
\label{alg:(Berlekamp-Welch-for-Rational}(Berlekamp-Welch for Rational
Functions) Given $(\theta_{1},f_{1})$, $(\theta_{2},f_{2})$, ...,
$(\theta_{n},f_{n})$, find a rational function $F(\theta)$ of degree
$(k_{1},k_{2})$ exactly by evaluating it at $n>k_{1}+k_{2}+2t$ points
despite $t$ errors in the evaluation points:
\[
|\;\{i\in[n]\;|\;F(\theta_{i})\ne f_{i}\}\;|\le t.
\]
\end{alg}
\begin{proof}
The error polynomial by Definition~\ref{Def: ErrorPolynomial} satisfies
\begin{equation}
E(\theta_{i})F(\theta_{i})=E(\theta_{i})f_{i}.\label{eq:ErrorPoly}
\end{equation}
Let $W(\theta_{i})\equiv E(\theta_{i})f_{i}$, which implies that
$f_{i}=W(\theta_{i})/E(\theta_{i})$. Since $W(\theta)=E(\theta)F(\theta)$
is a $(k_{1}+t,k_{2})$ rational function, by Eq.~\eqref{eq:ErrorPoly},
$f_{i}$ is a $(k_{1}+t,k_{2}+t)$ rational function of $\theta$.
By Lemma~\ref{fact:Rational-function}, the linear system defined
by Eq.~\eqref{eq:ErrorPoly}, has a solution as long as $n>k_{1}+k_{2}+2t$.
If $W(\theta)/E(\theta)$ results in a rational function of degree
$(k_{1},k_{2})$ we are done and we simply output it as $F(\theta)$,
otherwise we decide that there were too many errors.

Can the algorithm find distinct $(E_{1}(\theta),W_{1}(\theta))$ and
$(E_{2}(\theta),W_{2}(\theta))$? We now show that $W_{1}(\theta)/E_{1}(\theta)$
and $W_{2}(\theta)/E_{2}(\theta)$ are equal, which means that $F(\theta)$
is learned uniquely even if there are multiple solutions. We have
\[
\frac{W_{1}(\theta)}{E_{1}(\theta)}=\frac{W_{2}(\theta)}{E_{2}(\theta)}\iff E_{1}(\theta)W_{2}(\theta)=E_{2}(\theta)W_{1}(\theta).
\]
Recall that both sides are bounded degree rational functions (of degree
$(k_{1}+2t,k_{2})$). So by evaluating them at enough points we can
determine them uniquely (Lemma~\ref{fact:Rational-function}). Since
at every $\theta_{i}$
\[
E_{1}(\theta_{i})f_{i}=W_{1}(\theta_{i}),\quad E_{2}(\theta_{i})f_{i}=W_{2}(\theta_{i}),
\]
solving for $f_{i}$, we have $E_{1}(\theta_{i})W_{2}(\theta_{i})f_{i}=f_{i}E_{2}(\theta_{i})W_{1}(\theta_{i})$
at every $\theta_{i}$. If $f_{i}=0$, then by Eq.~\eqref{eq:ErrorPoly}
$W(\theta_{i})=0$, otherwise we just cancel the $f_{i}$. This proves
the claim that $F(\theta)=E_{1}(\theta)/W_{1}(\theta)=E_{2}(\theta)/W_{2}(\theta)$.
Since $n>k_{1}+k_{2}+2t$, we are guaranteed to have enough points. 
\end{proof}
\begin{rem}
Here we gave an algorithm in which the number of errors $t<(n-k_{1}-k_{2})/2$. As in standard BW algorithm, it is entirely possible to find algorithms
that can handle more errors. The above is sufficient for our purposes
so we leave such refinements for future work. 
\end{rem}

\section{\label{sec:Close_to_Haar} Total variation distance from the Haar measure}

In this section, after introducing the Haar measure, in Lemma~\ref{lem:TVD_Haar}
we will prove that if $H=f(h)$ is a finite dimensional matrix drawn
from the Haar measure, then $Cf(\theta h)$ is $\Delta-$close in
total variational distance (TVD) to the Haar measure for $|1-\theta|\le\Delta\ll1$.
This lemma then will directly apply to each gate $C_{k}(\theta)$
and will be shown to prove $O(m\Delta)$ TVD of the full circuit $C(\theta)$
from $\mathcal{H}_{\mathcal{A}}$.

Recall that $\mathbb{O}(N)$, and $\mathbb{U}(N)$ denote the set
of orthogonal and unitary matrices respectively~\cite{movassagh2018efficient}.
The entries of these matrices are real ($\beta=1$), and complex ($\beta=2$)
respectively. In the special case that the determinant is equal to
one, these are denoted by $\mathbb{SO}(N)$ and $\mathbb{SU}(N)$.
If $G$ is any one of the matrix groups, then a \textit{uniform random
element} of $G$ is a matrix $V\in G$ whose distribution is\textit{
translation invariant}, which is called the Haar measure. This means that for any fixed $M\in G$,
\[
VM\stackrel{d}{=}MV\stackrel{d}{=}V,
\]
where $\stackrel{d}{=}$ is equality in the distribution sense. We
have the well-known theorem of Haar:
\begin{thm*}\label{Thm:Haar_Standard}
(Haar~\cite{haar1933massbegriff}) Let $G$ be any of $\mathbb{O}(N)$, $\mathbb{SO}(N)$,
$\mathbb{U}(N)$ or $\mathbb{SU}(N)$. Then there is a unique translation-invariant
probability measure on G.
\end{thm*}
Since $H_k$ is independent of $C_{k}$, this Theorem implies that $C_{k}H_{k}$
is also Haar if $H_{k}$ is. So below we can focus on the distribution
of $f(\theta h)$ as compared with Haar. Recall that the Cayley function
and its inverse are
\[
f(x)=\frac{1+ix}{1-ix}\qquad;\qquad f^{-1}(x)=i\frac{1-x}{1+x}\quad,
\]
where $f(-\infty)=-1.$ We wish to calculate the TVD between the Haar
distribution induced by the unitary matrix $f(h)$ where as before
we have that $h=h^{\dagger}$, and the distribution induced by $f(\theta h)$
with $\theta\ll1$. In practice $\theta$ is taken to be smaller than
any constant $\theta=o(1)$.

We proceed by spectral decomposition of the unitary matrix $H$,
\[
H=f(h)=\sum_{\alpha=1}^{N}e^{ir_{\alpha}}|\psi_{\alpha}\rangle\langle\psi_{\alpha}|\quad,
\]
where we identified the eigenvalues $f(h_{\alpha})$ on the unit circle
with the phases $e^{ir_{\alpha}}$ with $r_{\alpha}\in[-\pi,\pi)$.
Solving $h_{\alpha}=f^{-1}(e^{ir_{\alpha}})$ we find that $h_{\alpha}=\tan(r_{\alpha}/2)$;
as expected the eigenvalues $h_{\alpha}$ that tend towards infinity
in magnitude correspond to phases $r_{\alpha}$ near $\pm\pi$.

Similarly, define the eigenvalues of $f(\theta h)$ to be the phases
$e^{i\nu_{\alpha}}$, and we have
\[
f(\theta h)=\sum_{\alpha=1}^{N}f(\theta h_{\alpha})|\psi_{\alpha}\rangle\langle\psi_{\alpha}|\equiv\sum_{\alpha=1}^{N}e^{i\nu_{\alpha}}|\psi_{\alpha}\rangle\langle\psi_{\alpha}|\quad.
\]
We solve $\theta h_{\alpha}=f^{-1}(e^{i\nu_{\alpha}})=\tan(\nu_{\alpha}/2)$
for $\nu_{\alpha}$ to find
\begin{equation}
\nu_{\alpha}(r_{\alpha})=2\arctan(\theta\tan(r_{\alpha}/2))\quad.\label{eq:nu_alpha}
\end{equation}
Comment: At $\theta=0$ all $e^{i\nu_{\alpha}}=1$, which corresponds
to $Cf(0)=CI=C$, and at $\theta=1$ we have $e^{i\nu_{\alpha}}=e^{ir_{\alpha}}$,
which corresponds to $Cf(h)=CH$ which is Haar distributed by the
above Theorem.

Total variation distance (TVD) between the continuous probability
distributions $\mu$ and $\rho$ is
\begin{eqnarray}
\left\Vert \mu-\nu\right\Vert _{\text{TVD}} & \equiv & \frac{1}{2}\int_{-\pi}^{\pi}dx_{1}\cdots\int_{-\pi}^{\pi}dx_{N}\text{ }|\mu(\boldsymbol{x})-\rho(\boldsymbol{x})|\quad,\label{eq:TVD}
\end{eqnarray}
The key lemma is:
\begin{lem}
\label{lem:TVD_Haar}If $H=f(h)$ is a unitary matrix distributed
according to the Haar measure, then the distribution over unitaries
$f(\theta h)$ for $|1-\theta|\le\Delta\ll1$ is $O(\Delta)$-close
to Haar in total variation distance. 
\end{lem}
\begin{proof}
Denoting by $\boldsymbol{r}=(r_{1},\dots,r_{N})$, the distribution
over the eigenvalues of $H=f(h)$ is given by Weyl's seminar work
\cite{weyl1964symetrie},
\begin{equation}
d\mu(\boldsymbol{r})=\frac{(N!)^{-1}}{(2\pi)^{N}}\prod_{\alpha<\beta}|e^{ir_{\alpha}}-e^{ir_{\beta}}|^{2}\text{ }dr_{1}dr_{2}\cdots dr_{N}\label{eq:dMU-1}
\end{equation}
The distribution induced by $f(\theta h)$ is obtained by the change
of variables $r_{\alpha}\mapsto\nu_{\alpha}$ (see Eq.~\eqref{eq:nu_alpha}).
The corresponding distribution over $\boldsymbol{\nu}$ is denoted
by $\rho(\boldsymbol{\nu})$, whose density is
\begin{equation}
d\rho(\boldsymbol{\nu})=\frac{(N!)^{-1}}{(2\pi)^{N}}\prod_{\alpha<\beta}|e^{ir_{\alpha}(\nu_{\alpha})}-e^{ir_{\beta}(\nu_{\beta})}|^{2}\text{ }|\det J(\boldsymbol{\nu})|\text{ }d\nu_{1}d\nu_{2}\cdots d\nu_{N}\label{eq:dRHO-1}
\end{equation}
where $J(\boldsymbol{\nu})$ is the Jacobian of the matrix of transformation.
Since the change of variables $\nu_{\alpha}=2\arctan(\theta\tan(r_{\alpha}/2))$
results in a diagonal Jacobian matrix $J=\text{diag}(J_{1},\dots,J_{N})$,
we have
\[
|\det J(\boldsymbol{\nu})|=\prod_{\alpha=1}^{N}\left|\frac{\partial r_{\alpha}(\nu_{\alpha})}{\partial\nu_{\alpha}}\right|=\prod_{\alpha=1}^{N}\frac{1+\theta^{2}+\cos(r_{\alpha})(1-\theta^{2})}{2\theta}
\]
Jacobian is a local property and for $\theta\approx1$ it is very
well-behaved. Let $\Delta\ll1$ and the condition $|1-\theta|\le\Delta$
means $1-\Delta\le\theta\le1+\Delta$. The upper-bound on the Jacobian
of transformation is (since $N$ is a constant)
\[
|\det J(\boldsymbol{\nu})|\le\prod_{\alpha\in[N]}\frac{1+\Delta(1+\cos(r_{\alpha}))}{1-\Delta}=1+O(\Delta).
\]
Near $\theta=1$ we can let $\theta=1+\delta$ where $|\delta|\le\Delta$
and we have $\nu_{\alpha}(r_{\alpha})=2\arctan(\theta\tan(r_{\alpha}/2))=r_{\alpha}+\delta\sin(r_{\alpha})+O(\delta^{2})$.
Therefore,
\[
|e^{ir_{\alpha}(\nu_{\alpha})}-e^{ir_{\beta}(\nu_{\beta})}|^{2}\approx|e^{ir_{\alpha}}-e^{ir_{\beta}}|^{2}+O(\Delta).
\]
This along with the bound on Jacobian proves the claim  $\left\Vert \mu-\nu\right\Vert _{\text{TVD}}=O(\Delta)$.
\end{proof}

\section{\label{sec:RCS} Average-case hardness of Random Circuit Sampling}

We now turn our attention to the quantum complexity theory and the
application of the above for proving the hardness of sampling from
the output distribution of generic quantum circuits. Theorems~\ref{thm:MainResult-1}
and~\ref{thm:Robustness} together will prove the main theorem in
the manuscript. 

One can make formal the definition of a circuit architecture $\mathcal{A}$
(e.g.,~\cite{napp2019efficient}), however, it simply means the location
and layout of the circuit before specifying the actual local unitaries
(Fig.~\ref{fig:Architecture_C}). The architecture specifies on
which qubit(s) each $C_{k}$ applies and specifies the temporal order
of the application of the gates. The quantum circuit is then denoted
by $C_{{\cal A}}$, which for notational simplicity we denote by $C$.

One says that the circuit $C$ is \textit{generic} with respect to
the architecture ${\cal A}$ if the local unitaries $C_{k}$ are drawn
independently from the Haar measure.

\subsection{Formal results for the exact average-case hardness}

For any circuit $C$, one can insert a complete set of basis between
each $C_{k}$ and $C_{k+1}$ and represent the circuit in what is
at times called ``Feynman path integral'' form. The amplitude corresponding
to the initial state $|y_{0}\rangle$ and final state $|y_{m}\rangle$
is
\begin{equation}
\langle y_{m}|C|y_{0}\rangle=\sum_{y_{1},y_{2},\dots,y_{m-1}\in\{0,1\}^{n}}\langle y_{m}|\mathcal{C}_{m}|y_{m-1}\rangle\langle y_{m-1}|\mathcal{C}_{m-1}|y_{m-2}\rangle\cdots\langle y_{1}|\mathcal{C}_{1}|y_{0}\rangle.\label{eq:Feynman}
\end{equation}

\begin{dfn}
(Haar random circuit distribution) Let ${\cal A}$ be an architecture
over circuits and let ${\cal H}_{{\cal A}}$ be the distribution over
circuits in ${\cal A}$ whose local gates, denoted by $H_{k}$, are
independently drawn from the Haar measure. 
\end{dfn}
The random circuit sampling is then the following task: 
\begin{dfn}
\label{def:Random-Circuit-Sampling}(Random Circuit Sampling (RCS))
Given an architecture $\mathcal{A}$, the description of a circuit
$C\in {\cal H_{A}}$, and parameters $\epsilon>0$ and $\delta>0$,
sample from the output probability distribution induced by $C$ with
probability $1-\delta$ over the choice of $C$. That is draw $y\in\{0,1\}^{n}$
with probability $\text{Pr}(y)=|\langle y|C|0\rangle|^{2}$ up to
a total variation distance $\epsilon$ in time $\text{poly}(n,1/\epsilon)$. 
\end{dfn}
In RCS one seeks estimations of $|\langle y|C|0^{n}\rangle|^{2}$
but any bit string $|y\rangle$ is obtained by applying Pauli $X$
matrices to positions in $|0^{n}\rangle$ that correspond to 1's.
By the so called 'hiding property'~\cite{aaronson2011computational},
which guarantees an equality of probabilities, it is \textit{sufficient}
to prove the hardness of computing
\begin{equation}
\text{p}_{0}(C)\equiv|\langle0^{n}|C|0^{n}\rangle|^{2}.\label{eq:p0}
\end{equation}

As stated in the manuscript, there exist local quantum circuits with
$n$ qubits whose probability amplitudes are $\#P$-Hard to estimate
to within $1/poly(n)$ multiplicative error~\cite{bremner2011classical,bremner2011classical,terhal2002adaptive}.
The specification of the worst-case circuits is not relevant for our
purposes, we make use of their existence only. The task is to extend
this hardness to generic circuits with the same architecture. 

Informally, the quantum supremacy conjecture in the context of RCS
states: Approximating most amplitudes to $O(2^{-n}/\text{poly}(n))$
additive error is a $\#P-$Hard problem for most quantum circuits.
Formally, it reads

\begin{conj}
\label{QSupremacyConj} (Quantum Supremacy Conjecture~\cite{aaronson2011computational})
There is no classical randomized algorithm that performs RCS in time
$\text{poly}(n,\epsilon^{-1})$ where $\epsilon$ is the total variation
distance error. 
\end{conj}
The informal and the formal definition are related via a celebrated
algorithm of Stockmeyer for approximate counting. Suppose sampling
was efficient on a classical computer, then we could use Stockmeyer's
algorithm to approximate the amplitudes to $1/poly(n)$ multiplicative
error. This would collapse the polynomial hierarchy to the third level. This is a contradiction if we assume that the PH does not collapse.  Since the probability amplitudes
anti-concentrate for random circuits, the expectation value of an
amplitude is $2^{-n}$ and the $1/poly(n)$ multiplicative error is
equivalent to $2^{-n}/poly(n)$ additive error. The contrapositive
to this logic is that: Assuming polynomial Hierarchy does not collapse
to the third level, then proving hardness (indeed $\#P$-Hardness)
of approximating the probability amplitudes of the random circuits
to $2^{-n}/poly(n)$ additive error implies that sampling is inefficient
(i.e.,~hard) for any classical computer. 
\begin{dfn}
\label{def:(Deformed-Haar)}($\theta$-deformed Haar towards $C$)
Let ${\cal A}$ be the architecture of the worst case circuit $C=\mathcal{C}_{m}\,\mathcal{C}_{m-1}\cdots\mathcal{C}_{2}\,\mathcal{C}_{1}$.
Let $\theta\in[0,1]$ and define by $C(\theta)=\mathcal{C}_{m}(\theta)\,\mathcal{C}_{m-1}(\theta)\dots\mathcal{C}_{2}(\theta)\,\mathcal{C}_{1}(\theta)$,
where by Eq.~\eqref{eq:CayleyPath} we have $\mathcal{C}{}_{k}(\theta)=C_{k}(\theta)\otimes\mathbb{I}$
and $C_{k}(\theta)\equiv C_{k}f(\theta h_{k})$. This path is unitary
for all $\theta$ as depicted in Fig.~\ref{fig:Plot-of-Cayley}. Further,
each $f(h_{k})=H_{k}$ is a (local) unitary drawn independently from
the Haar measure and $C_{k}$ is the local unitary of the worst case
circuit (i.e.,~$\mathcal{C}_{k}=C_{k}\otimes\mathbb{I}$). We define
by ${\cal H}_{{\cal A},C,\theta}$ the distribution over $C(\theta)$. 
\end{dfn}
By the translation invariance property of the Haar measure (Section
\ref{sec:Close_to_Haar}), $C_{k}(1)$ implements a local unitary
from the Haar measure. Therefore, the distribution over $C(1)$ coincides
with ${\cal H}_{{\cal A}}$, and $\mathcal{C}_{k}(0)=\mathcal{C}_{k}=C_{k}\otimes\mathbb{I}$,
which is the fixed $k^{th}$ gate of the worst-case circuit $C$.
In summary, $C(\theta)\in{\cal H}_{{\cal A},C,\theta}$ with extremes
(See Fig.~\ref{fig:CircuitDeformation}):
\begin{figure}
\centering{}\includegraphics[scale=0.45]{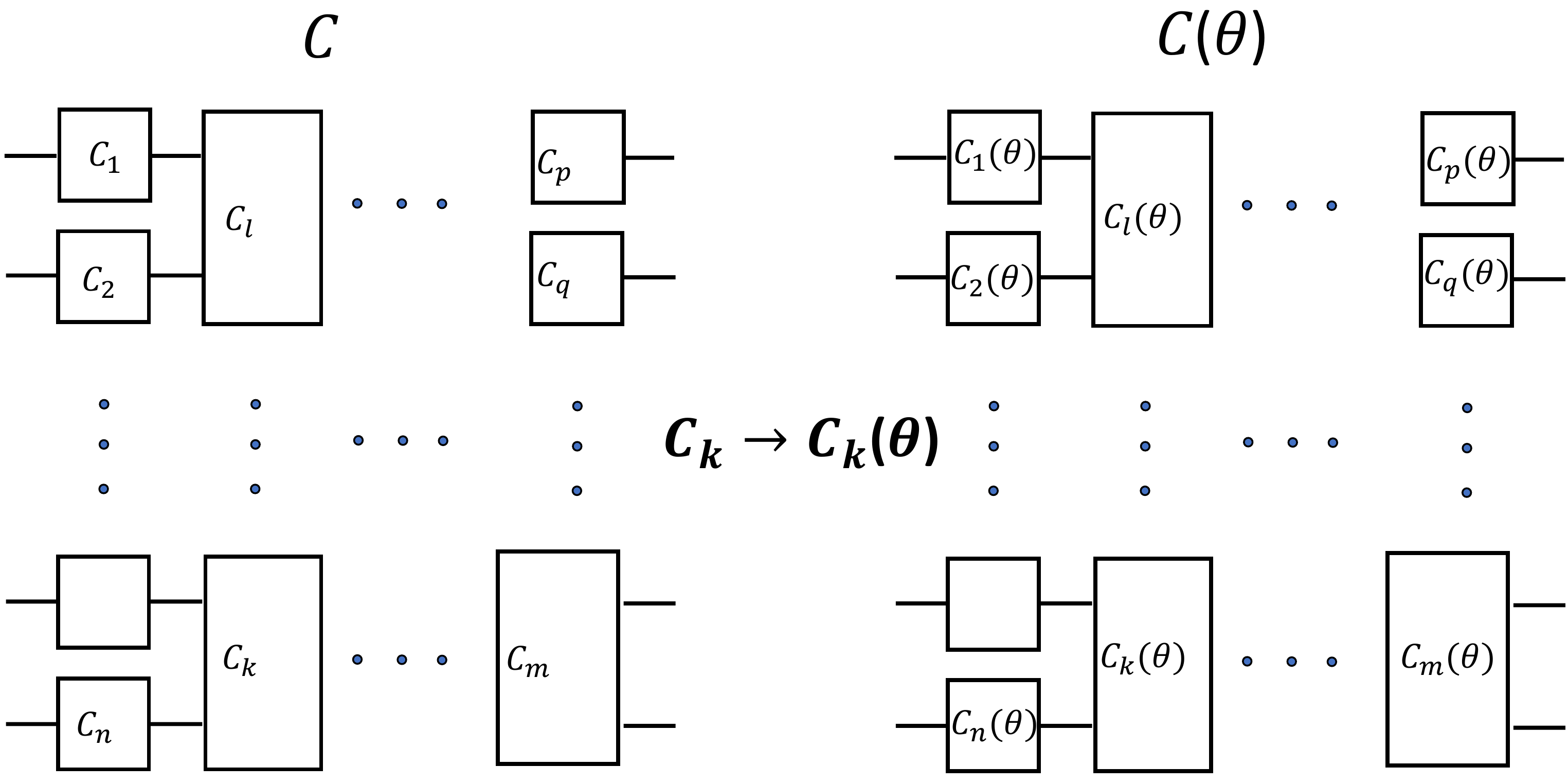}\caption{\label{fig:CircuitDeformation}Schematics of Definition~\ref{def:(Deformed-Haar)}:
The scrambling of the circuit $C$ to $C(\theta)$.}
\end{figure}
\begin{eqnarray}
\theta=0 & : & C_{k}(0)\quad\forall k\implies C(0)=C\quad\text{worst case circuit}.\label{Def:G}\\
\theta=1 & : & C_{k}(1)\quad\forall k\implies C(1)\in{\cal H}_{{\cal A}}\nonumber 
\end{eqnarray}
This naturally defines the deformation of Eq.~\eqref{eq:p0} via Eq.
\ref{eq:p0_theta-1}, which we recall is
\begin{equation}
p_{0}(\theta)\equiv|\langle0^{n}|C(\theta)|0^{n}\rangle|^{2},\label{eq:p0_theta}
\end{equation}
which at $\theta=1$ is the RCS problem and at $\theta=0$, we have
the \#$P$-Hard worst-case instance: $p_{0}(0)=p_{0}\equiv|\langle0^{n}|C|0^{n}\rangle|^{2}$.
\begin{lem}
\label{lem:mTheta}The total variation distance of ${\cal H_{A}}$
and ${\cal H}_{{\cal A},C,\theta}$ is $O(m\Delta)$ for $|1-\theta|\in\Delta\ll1$.
\end{lem}
\begin{proof}
By the translational invariance of Haar measure (Subsection~\ref{sec:Close_to_Haar}),
if $H_{k}$ is distributed according to the Haar measure then so is
$C_{k}H_{k}$ for any fixed $C_{k}$. Moreover the $\ell_{1}$ norm
that defines total variation distance is invariant under unitary multiplication.
So it suffices to compare the measures over ${\cal H}_{{\cal A},C,\theta}$
and ${\cal H}_{{\cal A}}$, which by Lemma~\ref{lem:TVD_Haar} have
TVD of $O(\Delta)$ over a single local gate. By the additivity of
TVD, the distribution induced by $C(\theta)$ which is denoted by
${\cal H}_{{\cal A},C,\theta}$ has a TVD from ${\cal H}_{{\cal A}}$
that is $O(m\Delta)$ for any $\theta$ satisfying $|1-\theta|\in\Delta\ll1$. 
\end{proof}
\begin{rem}
\label{rem:6} In an $n-$qubit circuit, $m=\Omega(\text{poly}(n))$; therefore,
if we take $|1-\theta|\le\Delta$ with $\Delta=O(1/\text{poly}(n))$
such that $\Delta=o(m^{-1})$, then we are guaranteed that the TVD between
${\cal H}_{{\cal A}}$ and ${\cal H}_{{\cal A},C,\theta}$ vanishes
with $n$. For example, in Google's experiment there are $n$ qubits
on a grid of size $\sqrt{n}\times\sqrt{n}$ and depth $\sqrt{n}$
resulting in $m=n^{3/2}$. Now if we take $\Delta=O(\frac{1}{n^{3/2}\log{n}})$,
the total variational distance to $\mathcal{H}_{\mathcal{A}}$ becomes
$O(m\Delta)\le O(\log^{-1}n)$ 
\end{rem}
From Eqs.~\eqref{eq:Cayley_path} and \eqref{eq:q_x} we have that
$\mathcal{C}_{k}(\theta)=C_{k}(\theta)\otimes\mathbb{I}.$ Therefore,
$\langle0^{n}|C(\theta)|0^{n}\rangle$ is equal to
\begin{eqnarray}
\langle0^{n}|C(\theta)|0^{n}\rangle & = & \langle0^{n}|\prod_{k=1}^{m}\mathcal{C}_{k}(\theta)|0^{n}\rangle  \equiv  \frac{\langle 0^n|P(\theta)|0^n\rangle}{Q(\theta)}\label{eq:P_over_Q},\qquad\text{where}\\
P(\theta)&\equiv&\sum_{\alpha_{1},\dots,\alpha_{m}=1}^{N}\prod_{k=1}^{m}p_{k,\alpha_{k}}(\theta)\text{ }\left(C_{k}|\psi_{k,\alpha_{k}}\rangle\langle\psi_{k,\alpha_{k}}|\otimes\mathbb{I}_{\hat{k}}\right)\label{eq:rationality_C}\\
Q(\theta) & \equiv & \prod_{k=1}^{m}q_{k}(\theta)=\prod_{k=1}^{m}\prod_{\alpha_{k}=1}^{N}(1-i\theta h_{k,\alpha_{k}})\label{eq:Q_Theta}\\
p_{k,\alpha_{k}}(\theta) & \equiv & (1+i\theta h_{k,\alpha_{k}})\prod_{\beta_{k}\in[N]/\alpha_{k}}(1-i\theta h_{k,\beta_{k}}),\nonumber 
\end{eqnarray}
where as before $h_{k,\alpha_{k}}$ and $|\psi_{k,\alpha_{k}}\rangle$
are the eigenpairs of the Hermitian matrix $h_{k}$, and $\mathbb{I}_{\hat{k}}$
denotes the trivial action of $C_{k}$ on all other qubits. $Q(\theta)$
is a polynomial of degree at most $Nm$ (recall that $N\in\{2,4\}$
for local quantum circuits). 

Since $\langle0^{n}|C(\theta)|0^{n}\rangle$ is a rational function
of degree at most $(4m,4m)$, we have that $p_{0}(\theta)$ as defined
by Eq.~\eqref{eq:p0_theta} is a rational function of degree at most
$(8m,8m)$. Moreover, from Eq.~\eqref{eq:rationality_C}, we have
\begin{equation}
p_0(P(\theta))\equiv\langle0^{n}|P(\theta)|0^{n}\rangle=Q(\theta)\text{ }\langle0^{n}|C(\theta)|0^{n}\rangle,\label{eq:D_theta}
\end{equation}
where $P(\theta)$ is a matrix whose entries are polynomials in
$\theta$. More importantly, $p_{0}(P(\theta))\equiv\langle0^{n}|P(\theta)|0^{n}\rangle$
is a polynomial of degree $Nm$ in $\theta$. 

\begin{rem}
The procedure is that we are given a fixed worst case circuit $C$
with the architecture ${\cal A}$ and whose $m$ local gates (i.e.,~$C_{k}$'s) are published. We then draw a corresponding set of $m$
local gates independently from the Haar measure (i.e.,~$H_{k}$'s)
and treat them as fixed. The latter is a realization of an average-case
circuit with architecture ${\cal A}$. Generating $H_{1},\dots,H_{m}$
takes $O(m)$ time as each can be generated from the QR decomposition
of a complex random gaussian matrix of size at most $N=4$~\cite{movassagh2018efficient}.
Therefore, the description of the quantum circuit is classically
efficient. We then choose a set of $\theta_{i}$ such that $|1-\theta_{i}|\in[0,\Delta]$. 
\end{rem}
\begin{thm}
\label{thm:MainResult-1}Let ${\cal A}$ be an architecture such that
computing $p_{0}=|\langle0^{n}|C|0^{n}\rangle|^{2}$ is $\#P-$Hard
in the worst case. Then, it is $\#P$-Hard to output $|\langle0^{n}|H|0^{n}\rangle|^{2}$
with the probability $\alpha=3/4+1/poly(n)$ over the choice of circuits $H\in\mathcal{H}_{\mathcal{A}}$.
\end{thm}
\begin{proof}
Let $C$ be an arbitrary circuit with architecture $\mathcal{A}$
and suppose we have at our disposal a classically efficient algorithm (an oracle) $\mathcal{O}$ such
that
\[
\underset{H\sim\mathcal{H}_{\mathcal{A}}}{\text{Pr}}[\mathcal{O}(H)=|\langle0^{n}|H|0^{n}\rangle|^{2}]\ge\alpha.
\]
Recall that under the Cayley path parametrization $p_{0}(\theta)=|\langle0^{n}|C(\theta)|0^{n}\rangle|^{2}$
is a rational function of degree at most $(8m,8m)$ such that $p_{0}(0)=|\langle0^{n}|C|0^{n}\rangle|^{2}$
and $p_{0}(1)=|\langle0^{n}|H|0^{n}\rangle|^{2}$ for some $H\in\mathcal{H}_{\mathcal{A}}$.
Divide the interval $I\equiv[1-\Delta,1+\Delta]$ into $L$ pieces
and for each $\theta_{i}\in I$ call the oracle $\mathcal{O}$ on
$C(\theta_i)$. Then using the generalized Berlekamp-Welch (Alg.~\ref{alg:(Berlekamp-Welch-for-Rational})
construct a $p'(\theta)$ such that
\[
p'(\theta_{i})=\mathcal{O}(C(\theta_{i}))\quad;\qquad|1-\theta_{i}|\le\Delta
\]
for $\alpha$ fraction of the $\theta_{i}\in I$. If the rational
function $p'_{0}(\theta)$ is not found output fail, otherwise output
$p'_{0}(0)$ as the proposed value of $p_{0}(0)=|\langle0^{n}|C|0^{n}\rangle|^{2}$. 

We now show that the above algorithm succeeds with sufficiently high
probability over the choice of $H$ and we can then repeat the algorithm
a small number of times (at most $poly(n)$) on different $H$'s and
output the majority result for $p_{0}(0)$. 

Since in Lemma~\ref{lem:mTheta} we proved that the total variation
distance between $\mathcal{H}_{\mathcal{A},\theta}$ and $\mathcal{H}_{\mathcal{A}}$
is $O(m\Delta)$ for $\theta\in I$ we have that
\begin{eqnarray*}
\text{Pr}[\mathcal{O}(C(\theta_{i}))=p_{0}(\theta_{i})\:|\:\theta_{i}\in I] & \ge & \alpha-||\mathcal{H}_{\mathcal{A}}-\mathcal{H}_{\mathcal{A},\theta,C}||_{TVD}\\
 & = & \alpha-O(m\Delta)\quad.
\end{eqnarray*}
where we wish to make $\alpha-O(m\Delta)\ge3/4+\delta/2$. We accomplish
this by taking $\alpha>3/4$, and $\Delta<O(m^{-1})$ as in Lemma
\ref{lem:mTheta}. %
For Berlekamp-Welch to succeed, $L$ needs to be bounded. Since in
Alg.~\ref{alg:(Berlekamp-Welch-for-Rational} $L>k_{1}+k_{2}+2t=8m+2t$
and from the above we have the error rate $t=(\frac{1}{2}-\delta)L$,
we conclude that $L>4m/\delta=poly(n)$. Let $\Theta$ be the set
of all $\theta_{i}\in I$ such that $\mathcal{O}(C(\theta_{i}))=p_{0}(\theta_{i})$;
 then 
$L-|\Theta|$ is the number of erroneous points. By Markov's inequality
\[
\text{Pr}\left[|\Theta|\ge(1+\delta)\frac{L}{2}\right]=1-\text{Pr}\left[(L-|\Theta|)\ge(1-\delta)\frac{L}{2}\right]\ge1-\frac{\frac{1}{4}-\frac{\delta}{2}}{\frac{1}{2}-\frac{\delta}{2}}\ge(1+\delta)/2
\]
We take $|\Theta|=(1+\delta)L/2$
in Berlekamp-Welch (Alg.~\ref{alg:(Berlekamp-Welch-for-Rational}),
which will succeed and outputs the rational function $p_{0}(\theta)$.
We are done because by repeating the call to the algorithm $O(\delta^{-2})=poly(n)$
times on different random circuits $C(\theta_{i})$, we can diminish
the error and take the majority to output $p_{0}(0)$. %
\end{proof}
We conclude that a polynomial time (efficient) classical computation of probability amplitudes implies $BPP=\#P$, which is believed to be highly unlikely.
\begin{rem}
 As remarked in the original BosonSampling
paper~\cite{aaronson2011computational}, it is entirely possible
that the above theorem for RCS may be strengthened to allow for an
oracle with the success probability $\alpha\ge1/\text{poly}(n)$ using
the results in~\cite{cai1999hardness}.
\end{rem}

\subsection{\label{subsec:Robustness}Proof of robustness}

The algebraic reduction goes as follows. For any given $\epsilon>0$,
suppose we can efficiently evaluate the tuples to within some additive
error for the random instances of the quantum circuit with high enough
probability. That is given a $\theta_{i}$ near $\theta=1$ we assume
that we have the estimate $p_{0}(\theta_{i})+\epsilon_{i}$, where
$\epsilon_{i}$ is the additive error. We want to maximize the additive
error tolerance and still be able to extrapolate to the worst-case
instance at $\theta=0$. That is, we wish to obtain the largest $\epsilon=\max_{i}|\epsilon_{i}|$
tolerance we can, with an eye on the fact that $\epsilon=2^{-n}/poly(n)$
proves the quantum supremacy conjecture! 

In order to do so,  we need to control the poles of the rational function
$p_{0}(\theta)$. It will be useful to transform the coordinate to
$\theta=1+z$ to make the problem symmetric. We write
\[
C_{k}(\theta=1+z)=\sum_{\alpha=1}^{N}\frac{1+i(1+z)h_{\alpha}}{1-i(1+z)h_{\alpha}}C_{k}|\psi_{k,\alpha}\rangle\langle\psi_{k,\alpha}|\quad,
\]
where ultimately $|z|\le\Delta=o(1)$. We will proceed to interpolate
between $z:\;|z|\le\Delta$ which corresponds to unitary matrices
with a TVD of at most $O(\Delta)$ from the Haar measure, and the
point $z=-1$ which coincides with the worst-case. The poles of the
rational function are the zeros of $Q(z)$, which can easily be seen
from Eq.\textasciitilde\ref{eq:Q_Theta} to be at $\theta=-i/h_{k,\alpha}$,
which in turn correspond to the roots $z=-1-i/h_{k,\alpha}$.

We have $1\pm i\theta h_{k,\alpha}=1\pm i(1+z)h_{k,\alpha}$. Let
$1\pm ih_{k,\alpha}=r_{k,\alpha}e^{\pm iu_{k,\alpha}}$ (parametrizes
a line with the real part one) where $r_{k,\alpha}=\sqrt{1+h_{k,\alpha}^{2}}$,
and $u_{k,\alpha}=\arctan\left(h_{k,\alpha}\right)$. Therefore
\[
1\pm i\theta h_{k,\alpha}=r_{k,\alpha}\left[e^{\pm iu_{k,\alpha}}\pm iz\frac{h_{k,\alpha}}{r_{k,\alpha}}\right].
\]
We rewrite $C_{k}(\theta)$ in Eq. (9) and cancel out the magnitudes
$r_{k,\alpha}$ to obtain
\begin{eqnarray*}
C_{k}(z) & = & \frac{\sum_{\alpha=1}^{N}r_{k,\alpha}\left[e^{iu_{k,\alpha_{k}}}+iz\frac{h_{k,\alpha}}{r_{k,\alpha}}\right]\prod_{\beta\in[N]\backslash\alpha}r_{k,\beta}\left[e^{-iu_{k,\beta}}-iz\frac{h_{k,\beta}}{r_{k,\beta}}\right]\;C_{k}|\psi_{k,\alpha}\rangle\langle\psi_{k,\alpha}|}{\prod_{\alpha=1}^{N}r_{k,\alpha}e^{-iu_{k,\alpha}}\left[1-iz\frac{h_{k,\alpha}}{r_{k,\alpha}}e^{iu_{k,\alpha}}\right]}\\
 & = & \frac{\sum_{\alpha=1}^{N}\left[e^{iu_{k,\alpha}}+iz\frac{h_{k,\alpha}}{r_{k,\alpha}}\right]\prod_{\beta\in[N]\backslash\alpha}\left[e^{-iu_{k,\beta}}-iz\frac{h_{k,\beta}}{r_{k,\beta}}\right]\;C_{k}|\psi_{k,\alpha}\rangle\langle\psi_{k,\alpha}|}{\prod_{\alpha=1}^{N}e^{-iu_{k,\alpha}}\left[1-iz\frac{h_{k,\alpha}}{r_{k,\alpha}}e^{iu_{k,\alpha}}\right]}
\end{eqnarray*}
Now the Eqs. (19)-(22) write under $\theta=1+z$
\begin{eqnarray}
|\langle0^{n}|C(z)|0^{n}\rangle|^{2} & = & |\langle0^{n}|\prod_{k=1}^{m}\mathcal{C}_{k}(z)|0^{n}\rangle|^{2}\equiv\frac{|\langle0^{n}|P(z)|0^{n}\rangle|^{2}}{|Q(z)|^{2}},\qquad\text{where}\label{eq:P_over_Q-1}\\
P(z) & \equiv & \sum_{\alpha_{1},\dots,\alpha_{m}=1}^{N}\prod_{k=1}^{m}p_{k,\alpha_{k}}(z)\text{ }\left[\left(C_{k}|\psi_{k,\alpha_{k}}\rangle\langle\psi_{k,\alpha_{k}}|\right)\otimes\mathbb{I}_{\hat{k}}\right]\label{eq:rationality_C-1}\\
|Q(z)|^{2} & \equiv & \prod_{k=1}^{m}q_{k}(z)=\prod_{k=1}^{m}\prod_{\alpha_{k}=1}^{N}\left|1-iz\frac{h_{k,\alpha_{k}}}{r_{k,\alpha_{k}}}e^{iu_{k,\alpha_{k}}}\right|^{2}.\label{eq:Q_Theta-1}\\
p_{k,\alpha_{k}}(z) & \equiv & \left[e^{iu_{k,\alpha_{k}}}+iz\frac{h_{k,\alpha_{k}}}{r_{k,\alpha_{k}}}\right]\prod_{\beta_{k}\in[N]\backslash\alpha_{k}}\left[e^{-iu_{k,\beta_{k}}}-iz\frac{h_{k,\beta_{k}}}{r_{k,\beta_{k}}}\right],\nonumber 
\end{eqnarray}

It is important to note that, Since $\frac{h_{k,\alpha_{k}}}{r_{k,\alpha_{k}}}<1$
and $z\in\Delta=o(m^{-1})$, the factors in $|Q(z)|^{2}$ are close
to one
\[
\left|1-iz\frac{h_{k,\alpha_{k}}}{r_{k,\alpha_{k}}}e^{iu_{k,\alpha_{k}}}\right|^{2}=1+(2z+z^{2})\frac{h_{k,\alpha_{k}}^{2}}{1+h_{k,\alpha_{k}}^{2}}\le1+O(\Delta),
\]
which results in
\[
|Q(z)|^{2}\le\prod_{k=1}^{m}\prod_{\alpha_{k}=1}^{N}\le1+O(Nm\Delta).
\]
Therefore, so long that we obey the bound $\Delta=o(m^{-1})$ obtained
in the TVD calculation above, we have $|Q(z)|^{2}\approx1$.\\

In the manuscript we mentioned that the worst-case hardness is to
within a constant multiplicative error~\cite{dyer2004relative}.
So the reduction provides a robustness with respect to additive errors,
so long as $\epsilon_{i}$ are small enough for the extrapolation
to $z=-1$ to be within the hardness interval that hugs $p_{0}(z=-1)$.
How do errors near $z=0$ amplify when extrapolated to $z=-1$?
To address this we first state Paturi's lemma:
\begin{lem*}\nonumber
(Paturi~\cite{paturi1992degree}) Let $p(z)$ be a polynomial of degree $d$, and suppose $|p(z)|\le\epsilon$
for $|z|\le\Delta$. Then $p(-1)\le\epsilon\exp[2d(1+\Delta^{-1})]$. 
\end{lem*}
Comment: Traditionally Paturi's lemma takes $|x|\le\Delta$ and bounds $p(+1)$. Let $z=-x$ and the above version is recovered.

Assume we have a classical algorithm ${\cal O}_{2}$ that to within
additive error $\epsilon=\max_{i}|\epsilon_{i}|$ has the property:
\[
\text{Pr}[|{\cal O}_{2}(C(z_{i}))-p_{0}(z_{i})|\le\epsilon]=1-1/\text{poly}(n)\:,\qquad |z_{i}|\le\Delta
\]
Namely, we are guaranteed to have a set of $(z_{i},|\langle 0^n|P(z_i)|0^n\rangle|^2+\epsilon_{i}|Q(z_{i})|^{2})$
with high probability. Now since degree of $|\langle 0^n|P(z)|0^n\rangle|^2$ is at most
$16m$, using Paturi's lemma, Rakhmanov's result, and the above bound on $|Q(z)|^2$
we have
\begin{equation}
\left||\langle 0^n|P_{exact}(-1)|0^n\rangle|^2-|\langle 0^n|P_{noisy}(-1)|0^n\rangle|^2\right|\le\epsilon(1+o(1))\;e^{32m(1+\Delta^{-1})}.\label{eq:Paturi}
\end{equation}
For example,, taking the parameters in Google's experiment,
we have $m=n^{3/2}$, we choose $\Delta=O(n^{-3/2}\log^{-2}(n))$
to ensure closeness in TVD as before (see Remark~\ref{rem:6}). 
\begin{rem}
In practice, we call the classical oracle some $poly(n)$ number of
times to obtain $(\theta_{i},p_{0}(\theta_{i}))$, which combined
with BW algorithm enables us to construct the rational function if
the rate of errors in evaluation is not too high. In using Paturi's
lemma there is a subtle question: Can the difference of the exact
and sampled polynomials be drastically different in $|1-\theta|\in[0,\Delta]$
despite agreeing well at the sampled points $(\theta_{i},p(\theta_{i}))$
(i.e.,~difference upper-bounded by $\max_{i}\epsilon_{i}|Q(\theta_{i})|^{2}$)?
This is not hard to remedy. If one samples $\theta_{i}$ uniformly
in the interval of length $\Delta$ near $\theta=1$, then by a theorem
due to Rakhmanov we are also guaranteed that the two polynomials are
close to one another everywhere in that interval~\cite{rakhmanov2007bounds}.
\end{rem}
\begin{thm}
\label{thm:Robustness}Assuming access to an oracle ${\cal O}_{2}$
as described above, it is $\#P$-Hard to compute $\text{p}_{0}(C(\theta))$
over ${\cal H}_{{\cal A}}$ to within $\epsilon=2^{ -\Omega\left(m^2\right)} $
additive error. 
\end{thm}
\begin{proof}
The proof follows from Eq.~\eqref{eq:Paturi}. From Dyer et al's result~\cite{dyer2004relative} we have hardness
guarantees to within constant multiplicative errors. We take $\epsilon=2^{-\Omega(n)}\exp(-O(m\Delta^{-1}))$ to ensure that the extrapolation to $z=-1$ lies to within $2^{-\Omega(n)}$ additive error of the $\#P$-Hard instance. The requirement of TVD implies that $\exp(-O(m\Delta^{-1}))=2^{ -\Omega\left(m^2\right)}$.
\end{proof}

Our scheme is resilient to noise $\epsilon=\exp(-\Omega\left(n^{3}\right))$
for parameters being used in the experimental setting of Google in
which the circuit is $\sqrt{n}\times\sqrt{n}\times\sqrt{n}$~\cite{arute2019quantum}.
And this scheme is resilient to noise of $\epsilon=\exp(-\Omega\left(n^{2}\right))$
for constant-depth proposals in which the geometry is $\sqrt{n}\times\sqrt{n}\times O(1)$
and therefore $m=O(n)$; these are classically hard to simulate
in the worst case~\cite{terhal2002adaptive} and recent numerical
advances seems to indicate that they are hard to simulate on average
for error thresholds similar to what we obtain and not those required
by the quantum supremacy conjecture (i.e.,~$\exp(-n)/\text{poly}(n)$)
\cite{napp2019efficient}. 

Note that the oracles ${\cal O}$ and ${\cal O}_{2}$ above do \textit{not}
need to succeed with probabilities $1/2+1/\text{poly}(n)$ and $1-1/\text{poly}(n)$
respectively over \textit{all} circuits with the given architecture.
For example, although ${\cal O}_{2}$ succeeds with probability close
to one, it is required to do so over circuits distributed close to
${\cal H}_{{\cal A}}$. 

In Lemmas such as Paturi's, the source of error $\epsilon$  is abstracted
away. It may be the sampling imprecisions, finite
precision in the inferred polynomial coefficients or what not.

\subsection{\label{subsec:Inadequacy-of-Taylor}Inadequacy of Taylor series truncation}

A proof of the average-case hardness of a non-unitary approximation
of RCS based on the truncation of the Taylor series expansion was
given by Bouland et al~\cite{bouland2018quantum}. In this section
and in what follows we adapt the setup and notation of Bouland et
al~\cite{bouland2018quantum}. Their reduction is fundamentally different
from ours. For example, among others, they follow the convention (e.g.,
in BosonSampling paper~\cite{aaronson2011computational}) and place
the worst-case instances at $\theta=1$ and average-case instances at $\theta\approx0$.
We do the opposite, which gives a path that is simpler and involves
the product of only two unitaries and not three; this also leads to
better TVD of the circuits from $\mathcal{H}_{\mathcal{A}}$.

In~\cite{bouland2018quantum} for each local gate $C_{j}$ of the
worst case circuit one picks a corresponding Haar matrix $H_{j}$
and forms the interpolating local gate $C_{j}(\theta)\equiv C_{j}H_{j}e^{-ih_{j}\theta}$,
where $\exp(-ih_{j})=H_{j}^{\dagger}$. The full unitary circuit is
$C(\theta)=C_{m}(\theta)C_{m-1}(\theta)\cdots C_{1}(\theta)$ and
$C(0)\in\mathcal{H}_{\mathcal{A}}$, and $C(1)=C$ is the worst case
circuit. However, the interpolation is only useful if it leads to
a low degree polynomial function, where the previous arsenal developed
for permanents in BosonSampling could be used~\cite{aaronson2011computational}.
To meet this requirement, they truncate the Taylor series expansion
of $\exp(-ih_{j}\theta)$ at the $K^{\text{th}}$ order~\cite{bouland2018quantum}.
Since the exponential function has an infinite power series, in order
to obtain a polynomial, they truncate the Taylor series expansion
of $\exp(-ih_{j}\theta)$ at the $K^{\text{th}}$ order
\begin{figure}
\centering{}\includegraphics[scale=0.5]{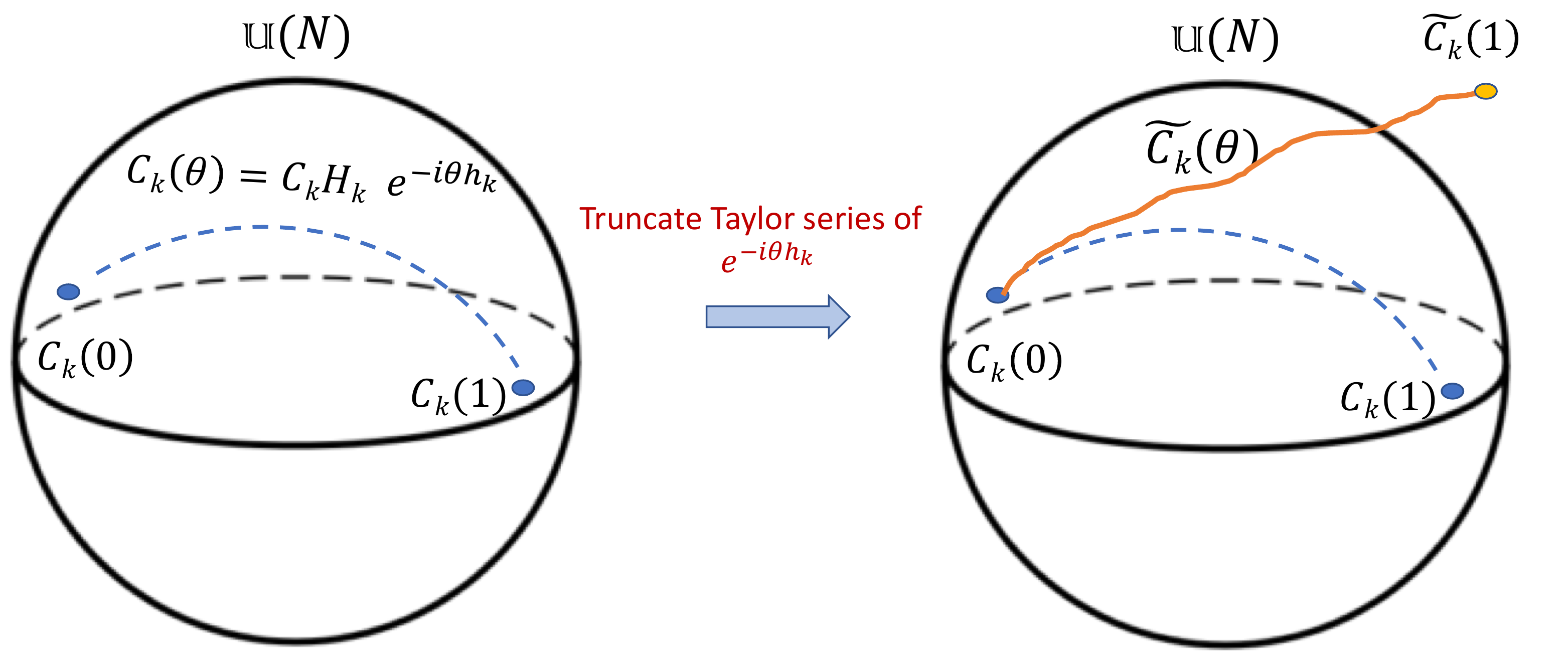}\caption{\label{fig:NonUnitary} Schematics of the deviation from the geodesic
and unitarity that results from the truncation of the Taylor series
(Eq.~\eqref{eq:Truncated}). Note that the path leaves the unitary
group as $\theta$ tends towards $\theta=1$.}
\end{figure}
\begin{equation}
\widetilde{C_{j}}(\theta)=C_{j}H_{j}\left\{ \sum_{k=0}^{K}\frac{(-ih_{j}\theta)^{k}}{k!}\right\} .\label{eq:Truncated}
\end{equation}
Note that the truncation leads to non-unitary local gates, and consequently
a non-unitary circuit. Therefore, to make the reduction of the worst-case
hardness to average work using Lipton's reduction (which is based
on polynomial reduction), they assume there exists a classical algorithm
$O$ that takes as input the non-unitary truncated description of
the circuit (compare with $O$ in Fig.~\ref{fig:Reduction-idea}).
They then succeed in proving the hardness based on this new complexity
theoretical assumption (see~\cite{napp2019efficient} for more discussion).
In summary, they need to \textit{assume an oracle with respect to
a non-unitary circuit}. Also a robustness of $\mathcal{O}(\exp(-\text{poly}(n)))$
with respect to additive error was claimed, which presupposes the
use of the, not so natural, non-unitary oracle access~\cite{bouland2018quantum}.

In order to reduce the complexity of the $\#P$-Hard problem to average
case, one needs to assume that there is an oracle that exactly computes
$\text{p}_{0}(C(0))$ where the local gates are Haar distributed.
The first issue with a non-unitary circuit is that this oracle cannot
be called. So they assume a different oracle that exactly computes
$\text{p}_{0}(\widetilde{C}(0))$. Then the claim is that the extrapolations
(i.e.,~$\text{p}_{0}(\widetilde{C}(0))$ ) is sufficiently close to
$\text{p}_{0}(C)$. This easily leads to the bound (as shown in~\cite{bouland2018quantum}),
\begin{equation}
|\text{p}_{0}(\widetilde{C}(1))-\text{p}_{0}(C)|\le\frac{2^{O(mn)}}{K!}\approx e^{O(mn-K\ln K)}.\label{eq:TrucationError}
\end{equation}

They use the above construction and to obtain a robustness with respect
to noise of $O(\exp(-\text{poly}(n)))$. Their robustness proof relies
on Paturi's lemma~\cite{paturi1992degree} and Rakhmanov's bounds
\cite{rakhmanov2007bounds}. Let the polynomial $p(\theta)=\text{p}_{0}(\widetilde{C}(\theta))-\text{p}_{0}(C(\theta))$,
then Paturi's lemma says 
\begin{lem*}
(Paturi) Let $p(\theta)$ be a polynomial of degree $d$, and suppose
$|p(\theta)|\le\epsilon$ for $|x|\le\Delta$. Then $p(1)\le\epsilon\exp[2d(1+\Delta^{-1})]$. 
\end{lem*}
In general, the \textit{robustness} claims correspond to the supremum
of $\epsilon$.

The robustness claims would be fine for the truncated circuits if
they also applied to an actual (i.e.,~unitary) circuit. This means
that if we call a standard oracle $O$ that takes in the exact unitary
gates as an input, then it would also give rise to additive errors
of $O(\exp(-\text{poly}(n)))$. Let us take the error to be only due
to truncation (Eq.~\eqref{eq:TrucationError}) and ignore all other
imperfections such as errors resulting from noisy polynomial sampling,
numerical round offs, or experimental limitations etc. Indeed one
can treat $K$ as a free variable and make it a sufficiently large
polynomial to compensate for $mn$ in the exponent of Eq.~\eqref{eq:TrucationError}.
This would lead to exponentially small errors in computing Eq.~\eqref{eq:TrucationError}
as stated in~\cite{bouland2018quantum}. Then, in order to sample
from distributions near the Haar measure, in Paturi's lemma they take
$\Delta=1/\text{poly}(n)$ as an \textit{independent} free variable.

However, the non-unitary approximation of the circuit is $\widetilde{C}(\theta)=\widetilde{\mathcal{C}}_{m}(\theta)\widetilde{\mathcal{C}}_{m-1}(\theta)\cdots\widetilde{\mathcal{C}}_{1}(\theta)$,
which has $m$ gates each truncated at $K^{th}$ order. Since the
degree of the products of polynomials is the sum of their degrees,
the entries of $C(\theta)$ become polynomials of degree $mK$; consequently
$|\langle0^{n}|\widetilde{C}(\theta)|0^{n}\rangle|^{2}|$ is a polynomial
of degree $2mK.$ Therefore, the degree in Paturi's lemma is $d=2mK$
and we have (by Eq.~\eqref{eq:TrucationError} and Paturi's lemma)
\begin{equation}
p(1)\le\exp[O(mn-K\ln K)]\exp[4mK(1+\Delta^{-1})].\label{eq:BlowUpTruncation}
\end{equation}
This is the fundamental observation: we see that $K$ and $d$ are
\textit{dependent}. For the errors not to blow up, one needs to take
$K\ge O(\exp(4m(1+\Delta^{-1})))=O(\exp(\text{poly}(n)))$ for the
right-hand side of Eq.~\eqref{eq:BlowUpTruncation} not to blow up-- the polynomial
would need to have an \textit{exponentially large degree in $n$,}
rendering a unitary oracle based scheme inefficient. This contradicts
the assumption of the existence of an efficient classical algorithm in
the reduction (c.f. Fig.~\ref{fig:Reduction-idea})

Ref.~\cite{bouland2018quantum} therefore, proceeds by assuming
a \textit{non-unitary} oracle, which for $\theta\ll1$ outputs $\text{p}_{0}(\widetilde{C}(\theta))$
exactly and with high probability. This is somewhat unnatural as the
oracle would take as inputs a classical description of a non-unitary
``circuit'' because of the truncation of the Taylor series to the
$K^{\text{th}}$ order (see Eq.~\eqref{eq:Truncated}). Under this
new complexity theoretical assumption, they succeed in proving that
approximating $\text{p}_{0}(\widetilde{C}(1))=|\langle0^{n}|\widetilde{C}(1)|0^{n}\rangle|^{2}$
to within $\exp(-\text{poly}(n))$ additive error is hard. 
\begin{figure}
\begin{centering}
\includegraphics[scale=0.4]{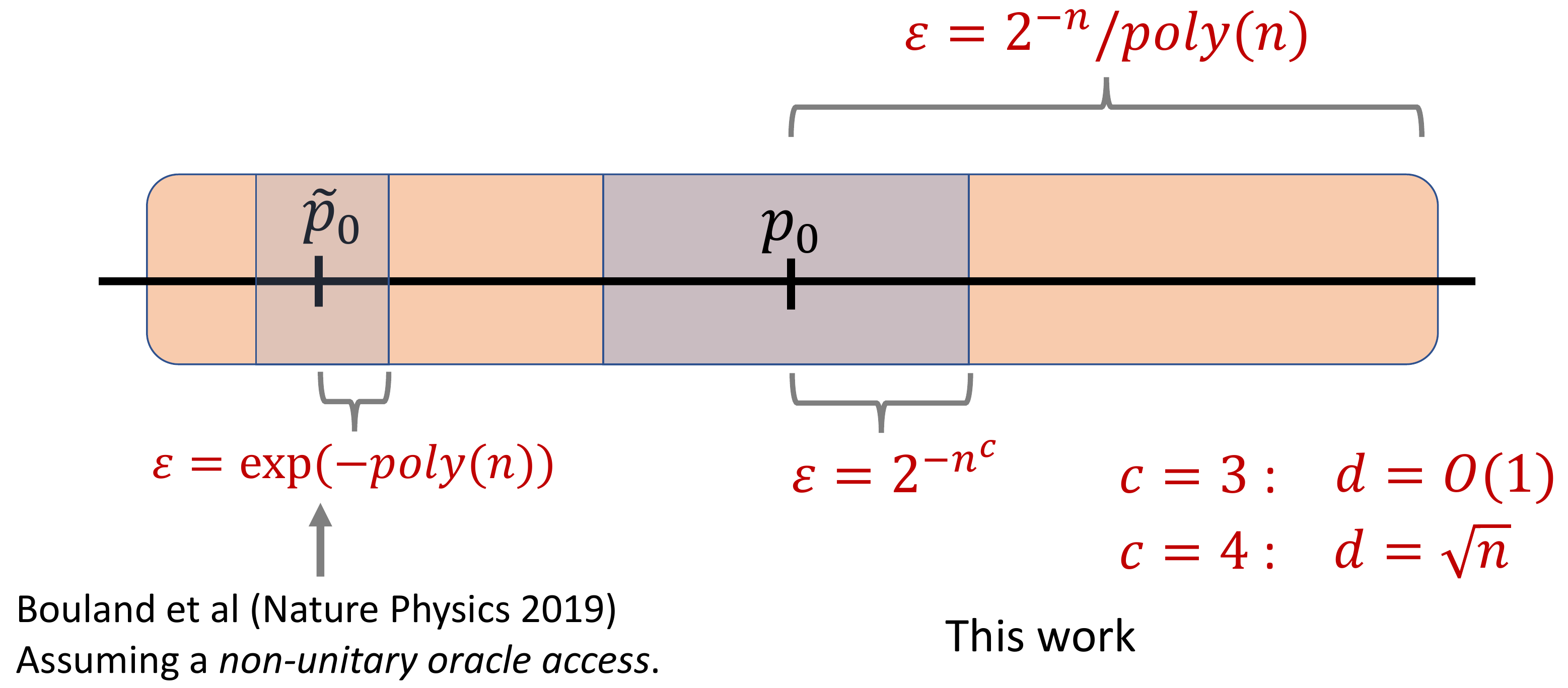}\caption{\label{fig:This-work-vs-Bouland}This work vs. the previous work~\cite{bouland2018quantum}
that assumes an access to a non-unitary oracle. Moreover, in this
section we showed that the non-unitary oracle does not provide robustness
if the input to it is the standard description of the (unitary) circuit. }
\par\end{centering}
\end{figure}

Bouland et al proved
that the TVD of the geodesic path, denoted by $\mathcal{H}_{\mathcal{A},\theta}$
is $O(m\theta)$ from $\mathcal{H}_{\mathcal{A}}$.~\cite{bouland2018quantum} . The truncated
path however is not a unitary and therefore has eigenvalues that do
not lie on the unit circle and Weyl's formula~\cite{weyl1964symetrie}
does not apply. Therefore, the analysis in~\cite{bouland2018quantum}
 does not readily extend in proving a small TVD between the distribution
over the non-unitary approximation of circuits obtained from truncation
of the Taylor series and $\mathcal{H}_{\mathcal{A}}$. But we like
to emphasize that in operator norm sense the truncation error obtained
from non-unitary approximation is indeed exponentially small. 

In summary their nice work makes an extra complexity theoretical assumption
that presupposes a non-unitarity oracle access. We summarize this
in comparison to our findings in Fig.~\ref{fig:This-work-vs-Bouland}.
We also note that any point $p$ in the interval $[p_{0}-2^{-n}/poly(n),p_{0}+2^{-n}/poly(n)]$
that satisfies exact
$\#P$-Hardness can be used to prove the supremacy conjecture if one
can prove a hardness for an interval encapsulating $p$ that is large
enough to encompass the supremacy interval. In this section we proved
that $\tilde{p}_{0}$, obtained from Taylor series truncation, will
not do. The reason is that the $\#P$-Hardness cannot be claimed with
a unitary oracle access.

\end{document}